\documentclass[preprint]{vldb}
\usepackage[utf8]{inputenc}
\usepackage{color}
\usepackage{amssymb}
\usepackage{amsmath}
\usepackage{graphicx, subfigure}
\usepackage{balance}
\usepackage{algorithm}
\usepackage{algorithmic}
\usepackage{footnote}
\usepackage{enumitem}
\usepackage{times}
\newtheorem{theorem}{Theorem}

\newtheorem{definition}{Definition}
\setlist[itemize]{noitemsep, topsep=4pt}
\title{Leveraging History for Faster Sampling of Online Social Networks}
\usepackage[bookmarks=false]{hyperref}
\numberofauthors{3}

\author{
\alignauthor Zhuojie Zhou\\
       \affaddr{George Washington University}\\
       \email{rexzhou@gwu.edu}
\alignauthor Nan Zhang\\
       \affaddr{George Washington University}\\
       \email{nzhang10@gwu.edu}
\alignauthor Gautam Das\\
       \affaddr{University of Texas at Arlington}\\
       \email{gdas@uta.edu}
}

\begin{document}
\maketitle

\begin{abstract}
With a vast amount of data available on online social networks, how to enable efficient analytics has been an increasingly important research problem. Many existing studies resort to sampling techniques that draw random nodes from an online social network through its restrictive web/API interface. While almost all of these techniques use the exact same underlying technique of {\em random walk} - a Markov Chain Monte Carlo based method that iteratively transits from one node to its random neighbor.

Random walk fits naturally with this problem because, for most online social networks, the only query we can issue through the interface is to retrieve the neighbors of a given node (i.e., no access to the full graph topology). A problem with random walks, however, is the ``burn-in'' period which requires a large number of transitions/queries before the sampling distribution converges to a stationary value that enables the drawing of samples in a statistically valid manner.

In this paper, we consider a novel problem of speeding up the fundamental design of random walks (i.e., reducing the number of queries it requires) {\em without} changing the stationary distribution it achieves - thereby enabling a more efficient ``drop-in'' replacement for existing sampling-based analytics techniques over online social networks. Technically, our main idea is to leverage the history of random walks to construct a higher-ordered Markov chain. We develop two algorithms, {\em Circulated Neighbors} and {\em Groupby Neighbors} Random Walk (CNRW and GNRW) and rigidly prove that, no matter what the social network topology is, CNRW and GNRW offer better efficiency than baseline random walks while achieving the same stationary distribution. We demonstrate through extensive experiments on real-world social networks and synthetic graphs the superiority of our techniques over the existing ones.
\end{abstract}

\section{Introduction}
\subsection{Motivation}

With the broad penetration of online social networks and the multitude of information they capture, how to enable a {\em third party}\footnote{i.e., one who is not the social network owner - examples include sociologists, economists, etc.} to perform efficient analytics of data available on social networks - specifically, to answer global and conditional aggregates such as SUM, AVG, and COUNT (e.g., the average friend count of all users living in Texas) - has become an increasing important research problem in the database community \cite{navlakha2008graph, 13275946, thirumuruganathan2014aggregate}. Applications of such aggregate estimations range from sociology research, understanding economic indicators to observing public health trends, bringing benefits to the entire society.

Technically, the challenge of data analytics over online social networks mainly stems from the limitations of their (web) query interfaces available to a third party - most online social networks only allow {\em local neighborhood queries}, with input being a node (i.e., user) and output being its immediate neighbors. Such a query interface makes retrieving the entire graph topology prohibitively expensive for a third party (in terms of query cost - due to the large size of real-world social networks). To address the problem, existing studies resort to a sampling approach, specifically the sampling of nodes (i.e., users) through the restrictive interface of a social network \cite{zhang2014exploration, zhou2013faster}, to enable aggregate estimations based on the sampled nodes. This is also the approach we focus on in the paper.

\subsection{Existing Techniques and Their Problems}

While there have been a wide variety of social network sampling designs \cite{Leskovec2006a, Kurant, 58571882, Gjoka2010} proposed in the literature, the vast majority of them share the same core design: {\em random walk}, a Markov Chain Monte Carlo (MCMC) method which iteratively transits from one node to a random neighbor. The variations differ in their design of the {transition probabilities}, i.e., the probability distributions for choosing the random neighbor, which are naturally adjusted for the various types of analytical tasks they target. Nonetheless, the core random walk design remains the same - after performing the random walk for a number of steps, the current node is taken as a sample - and, after repeated executions of random walks, we can generate (statistically accurate) aggregate estimation from the multiple sample nodes.

To understand the problem with this core design, it is important to observe the key bottleneck of its practical implementations. Note that such a sampling design incurs two types of overhead: One is query cost: Each transition in the random walk requires one {\em local neighborhood query} (described above) to be issued to the online social network, while almost all real-world social networks enforce rigid query-rate limits (e.g., Twitter allows only 15 local neighborhood queries every 15 minutes). The second type is the local processing (time and space) overhead for recording sampled nodes and computing aggregate estimations. One can see that the bottleneck here is clearly the query cost - compared with the extremely low query rate allowed by online social networks (e.g., 1 minute/query for Twitter), the local processing overhead (linear to the sample size \cite{li2014random}) is negligible. 

With the understanding that a sampling algorithm must minimize its query cost, the problem with existing techniques can be summarized in one sentence: The core random walk design requires a long ``burn-in'' before a node can be taken as a sample - since each step requires one query, this leads to a high query cost and therefore a very slow sampling process over real-world online social networks.

To understand what the burn-in period is and why it is required, we note that a sample node can only be used for analytics if we know the {\em bias} of the sampling process, i.e., the probability for a node to be taken as a sample. The reason for this requirement is simple - only with knowledge of the sampling bias can we properly correct it to ensure equal representation of all applicable tuples. Nonetheless, actually learning the sampling distribution without knowledge of the global graph topology is difficult. A desirable property of random walk is that it asymptotically converges (as it grows longer) to a {\em stationary} sampling distribution that can be derived without knowledge of the graph topology (e.g., with probability proportional to a node's degree for simple random walk \cite{spitzer1964principles}, or uniform distribution for Metropolis-Hastings random walk \cite{spitzer1964principles}). The number of steps required for a random walk to reach this stationary distribution is the ``burn-in'' period which, unfortunately, is often quite long for real-world social networks \cite{zhou2013faster, Gjoka2010}.

\subsection{Our Idea: History-Aware Random Walks}

The focus of this paper is to offer a {\em ``drop-in'' replacement} for this core design (of random walk), such that existing sampling-based analytics techniques over online social networks, no matter which analytics tasks they support or graph topologies they target, can have a better efficiency without changing other parts of their design. To do so, we shorten the burn-in period (i.e., reduce the query cost) of the fundamental random-walk design {\em without} changing the stationary distribution it achieves - ensuring the transparency of this change to how the core design is called upon in data analytics.

Technically, our key idea here is motivated by an observation on the potential waste of queries caused by the existing random walk design: Most existing random walk based techniques are Markov Chain Monte Carlo (MCMC) methods that are {\em memoryless} - i.e., they do not take into account the historic nodes encountered in a random walk in the design of future transitions. These methods, while simple, waste substantial chances for leveraging historic queries to speed up random walks - a waste we aim to eliminate with our proposed ideas.

Our main idea developed in the paper is to introduce historic dependency to the design of random walks over graphs. Specifically, we consider the nodes already visited by the current random walk while deciding where to go for the next step. We start with developing a simple algorithm called Circulated Neighbors Random Walk (CNRW). The difference between CNRW and the traditional random walk is sampling with and without replacement when deciding on the next move. Specifically, consider a random walk with the last move being $u \to v$. To determine the next move, the traditional random walk design is to sample uniformly at random from the neighbors of $v$. With CNRW, this sampling is done {\em with replacement} - in other words, if the random walk has transmitted through $u \to v \to w$ before, then we exclude $w$ from being considered for the next move, until the random walk has passed through $u \to v \to x$ for all neighbors $x$ of $v$. We prove that, while CNRW shares the exact same stationary distribution (see Definition~\ref{def:sta}) as the traditional simple random walk, it is provably more (or equally) efficient than the traditional random walk no matter what the underlying topology is.

A rationale behind the design of CNRW can be explained with the following simple example. Note that if a random walk over a large graph comes back to $u \to v$ after just passing though $u \to v \to w$, it is likely an indication that $v \to w$ leads to a small component that is not well connected to the rest of the graph (as otherwise the probability of going back to $u$ should be very small). Falling into such a small component ``trap'' is undesirable for a random walk, which needs to ``spread'' to the entire graph as quickly as possible. Thus, an intuitive idea to improve the efficiency of a random walk is to avoid following $v \to w$ again, so as to increase the chance of following the other (hopefully better-for-random-walk) edges associated with $v$. This intuition leads to the ``circulated'', without-replacement, design of CNRW\footnote{One might wonder why the circulation is conditioned upon traveling through an edge (i.e., $u \to v$) again rather than traveling through a vertex (say $v$) again - this is indeed a subtle point which we further address in Section 3.2.}.

Based on the idea of CNRW, we develop GroupBy Neighbors Random Walk (GNRW), which further improves the performance of a random walk by considering not only the nodes visited by a random walk, but also the observed attribute values of these visited nodes. To understand the key idea of GNRW, consider the following intuitive observation on the structure of an online social network: users with similar attribute values (e.g., age, interests, occupation) are more likely to be connected with each other. Leveraging this observation, to decide which neighbor of $v$ to transit to from $u \to v$, GNRW partitions all neighbors of $v$ into a number of strata according to their values on an attribute of interest (often the measure attribute to be aggregated - see discussions in Section~\ref{subsec:gnrw-idea}). Then, GNRW first selects a stratum uniformly at random without replacement - i.e., it ``circulates'' among the strata until selecting each stratum once. Then, GNRW chooses a neighbor of $v$ from the selected stratum, again uniformly at random without replacement. This neighbor then becomes the next step.

Similar to the theoretical analysis of CNRW, we also prove that GNRW shares the same stationary distribution as the traditional random walk. Intuitively, the design of GNRW aims to ``speed up'' the random walk by alternating between different attribute values faster, instead of getting ``stuck'' on a small component of the graph sharing the same (or similar) attribute value. We shall demonstrate through experimental evaluation that GNRW can significantly improve the efficiency of random walks, especially when the ultimate objective of these random walks is to support aggregate estimations on attributes used for stratification in GNRW.

\subsection{Summary of Contributions}

We summarize the main contributions as follows.

\begin{itemize}
\item We consider a novel problem of leveraging historic transitions to speed up random walks over online social networks.
\item We develop Algorithm CNRW which features a key change from traditional random walks: instead of selecting the next transition by sampling with replacement from all neighbors of the current nodes, CNRW performs the sampling {\em without} replacement, thus becoming a history-dependent random walk.
\item We develop Algorithm GNRW which further improves the efficiency of random walk by leveraging not only the history of nodes visited by a random walk, but also the observed attribute values of these visited nodes.
\item We present theoretical analysis which shows that, while CNRW and GNRW both produce samples of the exact same distribution as the traditional random walk, they are provably more efficient no matter what the underlying graph topology is.
\item Our contribution also includes extensive experiments over real-world online social networks which demonstrate the superiority of CNRW and GNRW over traditional random walk techniques.
\end{itemize}
\subsection{Paper Organization}

This paper is organized as follows. Section~\ref{sec:pre} describes preliminaries of random walks. Section~\ref{sec:cnrw} and section~\ref{sec:gnrw} introduce Circulated Neighbors Random Walk (CNRW) and Groupby Neighbors Random Walk (GNRW). Section~\ref{sec:discussion} is about two discussions of our algorithms. Section~\ref{sec:exp} shows the results of our experiments. Section~\ref{sec:ref} overviews the related works. Section~\ref{sec:con} provides our brief conclusion.
\section{Preliminaries}
\label{sec:pre}
In this section, we introduce preliminaries that are important for the technical development of this paper. Specifically, we start with introducing the access model for online social networks, followed by a discussion of the most popular method of sampling an online social network - random walks. Specifically, we shall introduce two important concepts related to random walks, {\em order} and {\em stationary distribution} of a Markov chain, which are critical for the development of our techniques later in the paper. Finally, at the end of this section, we define the key performance measures for a random walk based sampling algorithm.

\subsection{Access Model for Online Social Networks}
\label{subsec:model}

As a third party with no direct access to the backend data repository of an online social network, the only access channels we have over the data is the web and/or API interface provided by the online social network. While the design of such interfaces varies across different real-world online social networks, almost all of them support queries that take any user ID $u$ as input and return two types of information about $u$:
\begin{itemize}
\item $N(u)$, the set of all neighbors of $u$, and
\item all other attributes of $u$ (e.g., user self-description, profile, posts).
\end{itemize}
Note that the definition of a ``neighbor'' may differ for different types of social networks - e.g., Twitter distinguishes between followers and followees and thus features directed connections and asymmetric neighborhoods (i.e., a neighbor of $u$ may not have $u$ in its neighbor's list), while Google Plus features undirected ``friendship'' edges and thus symmetric neighborhoods. For the purpose of this paper, we consider undirected edges - i.e., $\forall v \in N(u)$, there is $u \in N(v)$. Note that online social networks that feature directed edges can be ``casted'' into this undirected definition. For example, we can define an undirected edge $e_{uv}$ if either edge $u\rightarrow v$ or $v\rightarrow u$ exists. Given the definition of node and neighbors, we consider the social network topology as an undirected graph $G(V,E)$. $V$ is the set of all the users (i.e. vertices/nodes) and $E$ is the set of all the connections between two users. $E=\{e_{uv}|u, v \in V\}$, where $e_{uv}$ is the connection between user $u$ and user $v$. Also, we use $k_v$ to denote the {\em degree} of node $v$, where $k_v = |N(v)|$. 

Many real-world online social networks also impose a {\em query rate limitation}. For example, Twitter recently update their API rate limits to ``15 calls every 15 minutes\footnote{https://dev.twitter.com/rest/public/rate-limiting}.'' Yelp offers only ``25,000 API calls per day\footnote{http://www.yelp.com/developers/faq}''.

\subsection{Random walk}
\label{subsec:randomwalk}

Almost all-existing works on sampling an online social network through its restrictive web/API interface (as modeled above) are variations of the random walk techniques. We briefly discuss the general concept of a random walk and a specific instance, the simple random walk, respectively as follows.

\subsubsection{Random Walk as an MCMC Process}

\vspace{1mm}
\noindent{\bf Order:} Intuitively, a random walk on an online social network randomly transits from one node to another according to a pre-determined randomized transition algorithm (and the neighborhood/node information it has retrieved through the restrictive interface of the online social network). From a mathematical standpoint, a random walk can be considered a Markov Chain Monte Carlo (MCMC) process with its {\em state} $X_i$ being the node visited by the random walk at Step $i$ and transition probability distribution
\begin{align*}
{} &\Pr(X_n=x_n\mid X_{n-1}=x_{n-1}, \dots , X_1=x_1) \\
=  &\Pr(X_n=x_n\mid X_{n-1}=x_{n-1}, \dots, X_{n-m}=x_{n-m})
\end{align*}
where $x_i \in V$. Here $m$ ($1 \leq m < n$) is the {\em order} of the Markov Chain. Most existing random walk techniques over online social networks are first-ordered (i.e., with $m = 1$). That is, with these random walks, the next transition only depends upon the current node being visited, and is independent of the previous history of the random walk. For higher-order random walks, the next transition is determined by not only the current node, but also an additional $m - 1$ steps into the history as well (e.g., the non-backtracking simple random walk (NB-SRW) \cite{58571882} has an order of $m = 2$, as it avoids transitioning back to the immediate last node whenever possible). Note that the random walk technique we will present in the paper has a much higher order than these existing techniques - more details in Section~\ref{sec:cnrw} and Section~\ref{sec:gnrw}.

\vspace{1mm}
\noindent{\bf Stationary Distribution:} The premise of using random walk to generate statistically valid samples is that, after performing a random walk for a sufficient number of steps, the probability for the random walk to land on each node converges to a {\em stationary distribution} $\pi$. In other words, after sufficient number of steps - $T$ steps, if the random walk has stationary distribution, then the probability of staying at node $u$ after $T$ steps is the same as after $T+1$ steps. The convergence has been proved to hold as long as the transition is {\em irreducible} and {\em positive recurrent} \cite{Neal2004}. Usually a Simple Random Walk (defined in section \ref{subsubsec:srw}) on a connected non-bipartite graph has a stationary distribution.

\begin{definition}\label{def:sta} {\bf [Stationary distribution].}
Let $\pi_j$ denote the {\em long run proportion of time that a random walk spends on node $j$} \cite{Neal2004}:
\begin{equation}
\pi_j = \lim_{n\rightarrow \infty}\frac{1}{n}\sum_{m=1}^nPr\{X_m=j\},
\label{eq:sta-dist}
\end{equation}
\end{definition}

One can see that the stationary distribution of a random walk determines how we can deal with the samples, and it enables us to further analyze the samples like estimating certain aggregates. 

\subsubsection{Simple Random Walk}
\label{subsubsec:srw}
In the current literature of sampling online social networks through their restrictive interfaces, Simple Random Walk (SRW) is one of the most popular techniques being used. Simple random walk is an order-1 Markov chain with transition distribution being uniform on all nodes in the neighborhood of the current node $x_i$. Formally, we have the following definition.

\begin{definition} {\bf [Simple Random Walk].}
Given graph $G(V,E)$, and a node $v\in V$, a random walk is called Simple Random Walk if it chooses uniformly at random a neighboring node $u\in N(v)$ and transit to $u$ in the next step.
\begin{equation}
    P_{vu}=\left\{
    \begin{array}{ll}
        1/k_v\,\,\,&\text{if $u \in N(v)$,}\\
        0\,\,\,&\text{otherwise.}
    \end{array}
    \right.
\end{equation}
\end{definition}

Corresponding to this transition design, one can easily compute the stationary distribution for simple random walk as
\begin{equation}
\pi_v = \frac{k_v}{2|E|}
\end{equation}

That is, SRW selects each node in the graph with probability proportional to its degree. In later part of the paper, we shall show how our proposed technique achieves the same sampling distribution while significantly improving the efficiency of random walks.

\subsection{Performance Measures for Sampling}
\label{subsec:measurements}

There are two important objectives for sampling from an online social network: minimizing the sampling bias and minimizing the query cost. We define the two performance measures corresponding to the two objectives respectively as follows.

\vspace{1mm}
\noindent{\bf Sampling Bias:} Intuitively, sampling bias measures the ``distance'' between the stationary distribution of a random walk and the actual sampling distribution achieved by the actual execution of the random walk algorithm. To understand why this is an important measure, we note that there is an inherent tradeoff between such a distance and the query cost required for sampling: If one does not care about the sampling bias and opts to stop a random walk right where it starts, the sampling bias will be the distance between the stationary distribution and a probability distribution vector of the form $\{1, 0, \ldots, 0\}$ (i.e., while the starting node has probability 100\% to be sampled, the other nodes have 0\%) - i.e., an extremely large sampling bias.

We propose to measure sampling bias with three metrics for various purposes in this paper: KL-divergence \cite{dodge2003oxford}, $\ell_2$-distance \cite{dodge2003oxford}, and a ``golden measure'' of the error of aggregate estimations produced by the sample nodes. Specifically, while the former two measures can be used in theoretical analysis and experimental analysis on small graphs, they are infeasible to compute for large graphs used in our experimental studies - motivating us to use the last measure as well. We shall further discuss the bias and the accuracy of various random walk algorithms in the experiments section \ref{sec:exp}.

\vspace{1mm}
\noindent{\bf Query Cost.} Another key performance measure for sampling over an online social network is the {\em query cost} - i.e., the number of queries (as defined in the access model described in Section~\ref{subsec:model}) one has to issue in order to obtain a sample node. Note that query cost here is defined as the number of {\em unique queries} required, as any duplicate query can be immediately retrieved from local cache without consuming the aforementioned query rate limit enforced by the online social network.

{\bf Asymptotic Variance.} Variance is one of the most important things to justify the performance of random walk algorithms. Of course, we can measure the sampling bias given certain query cost, but after sufficient number of steps, the efficiency of an estimator provided by random walk samples is directly tied to the variance, which is also important for us to theoretically compare the performance of random walks. For the purpose of this paper, we adopt the commonly used asymptotic variance in Markov Chain.
\begin{definition} {\bf [Asymptotic Variance].}
Given a function $f(\cdot)$ and an estimator $\hat{\mu}_n = \frac{1}{n}\sum_{t=1}^nf(X_t)$, the asymptotic variance of the estimator is defined as
\begin{equation}
V_{\infty}(\hat{\mu}) = \lim_{n\rightarrow \infty}n\text{Var}(\hat{\mu}_n)
\end{equation}
\end{definition}
Note that this does not depend on the initial distribution for $X_0$. In addition, in practice, we would use an estimator based on only $X_t$ with $t$ greater than some very large number - say $h$ - that we believe the chain has reached a distribution close to its stationary distribution $\pi$ after $h$ steps.

\section{CNRW: Circulated Neighbors \\ Random Walk}
\label{sec:cnrw}

\vspace{1mm}
In this section, we develop Circulated Neighbors Random Walk (CNRW), our first main idea for introducing historic dependency to the design of random walks over graphs. Specifically, we start with describing the key idea of CNRW, followed by a theoretical analysis of (1) its equivalence with simple random walk (SRW) in terms of (stationary) sampling distribution, and (2) its superiority over SRW on sampling efficiency. Finally, we present the pseudo code for Algorithm CNRW. 

\subsection{Main Idea and Justification}
\label{subsec:cnrw-idea}
\vspace{1mm}
\noindent{\bf Key Idea of CNRW:}
In traditional random walks, the transition at each node is {\em memoryless} - i.e., when the random walk arrives at a node, no matter where the walk comes from (i.e., what the incoming edge is) or which nodes the walk has visited, the outgoing edge is always chosen uniformly at random from all edges attached to the node. The key idea of CNRW is to replace such a memoryless transition to a stateful process. Specifically, given the previous transition of the random walk $u \to v$, instead of selecting the next node to visit by {\em sampling with replacement} from $N(v)$, i.e., the neighbors of $v$, we perform such sampling {\em without replacement}.

\begin{figure*}
\centering

  \subfigure[$N(v)$]{
  \includegraphics[width=.16\textwidth]{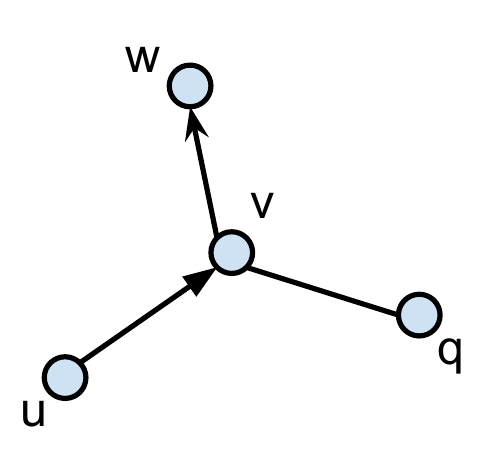}
  \label{fig:CNRW-demo1}
  }
    \subfigure[$N(v)-\{w\}$]{
  \includegraphics[width=.16\textwidth]{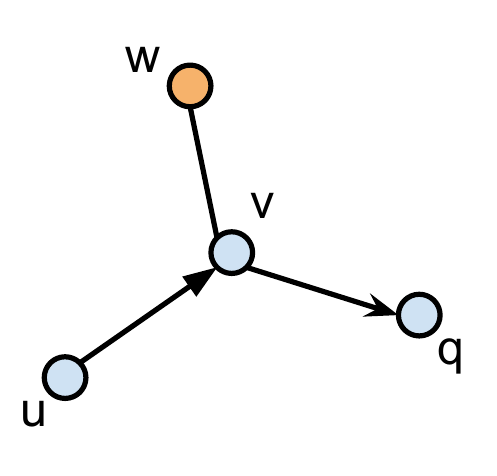}
  \label{fig:CNRW-demo2}
  }
    \subfigure[$N(v)-\{w, q\}$]{
  \includegraphics[width=.16\textwidth]{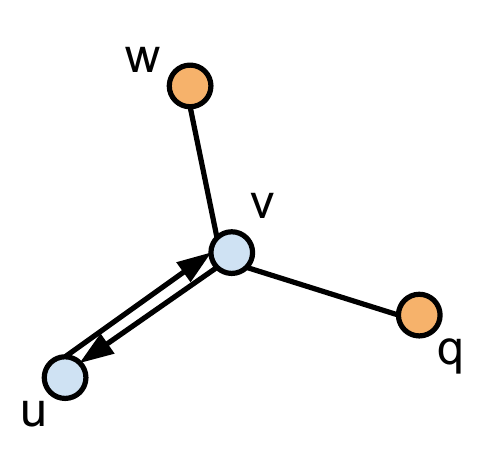}
  \label{fig:CNRW-demo3}
  }
    \subfigure[$N(v)-\{w, q, u\}$]{
  \includegraphics[width=.16\textwidth]{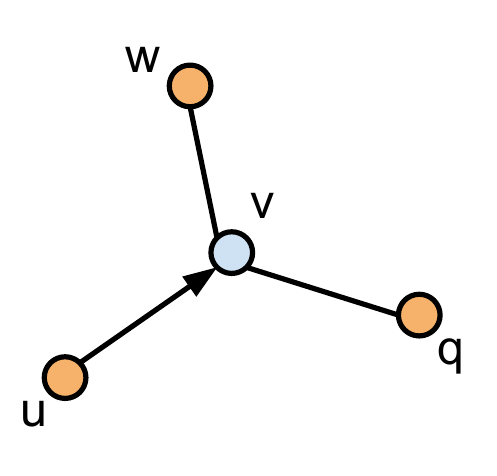}
  \label{fig:CNRW-demo4}
  }    
    \subfigure[$N(v)$]{
  \includegraphics[width=.16\textwidth]{pic/CNRW-demo1.pdf}
  \label{fig:CNRW-demo5}
  }
 \caption{Demo of CNRW, it chooses the next candidate from the set $N(v) - b(u,v)$.}
\label{fig:CNRW-demo}
\end{figure*}

Such a change is demonstrated through an example in Figure \ref{fig:CNRW-demo}: When a CNRW transits through $u \to v$ for the first time, it selects the next node to visit in the same way as traditional random walk - i.e., by choosing $w$ uniformly at random from $N(v)$. Nonetheless, if the random walk transits through $u \to v$ again in the future, instead of selecting the next node from $N(v)$, we limit the choice to be from $N(v) - \{w\}$. One can see that, given a transition $u \to v$, our selection of the next node to visit is a process of {\em sampling without replacement} from $N(v)$. This, of course, continues until $\forall w \in N(v)$, the random walk has passed through $u \to v \to w$, at which time we reset memory and restart the process of sampling without replacement. Also, we introduce a notation $b(u,v)$, which is defined as a set of nodes in $N(v)$ that we have passed through. Thus, we generalize the idea of CNRW as:
\begin{enumerate}
\item Each time when the random walk travels from $u$ to $v$, we uniformly choose the next candidate node w from $N(v)-b(u,v)$.
\item Let $b(u,v) \leftarrow b(u,v)\cup \{w\}$. If $b(u,v) = N(v)$, let $b(u,v) = \emptyset$.
\end{enumerate}

\noindent{\bf Intuitive Justification:} To understand why CNRW improves the efficiency of the sampling process, we start with an intuitive explanation before presenting in the next subsection a rigid theoretical proof of CNRW's superiority over the traditional simple random walk (SRW) algorithm.

Intuitively, the justification for CNRW is quite straightforward: If a random walk travels back to $u \to v$ after only a small number of steps, it means that the choice of $v \to w$ is not a ``good'' one for sampling because, ideally, we would like the random walk path to ``propagate'' to all parts of graph as quickly as possible, instead of being stuck at a small, strongly connected subgraph like one that includes $u$, $v$ and $w$. As such, it is a natural choice for CNRW to avoid following $v \to w$ again in the next transition, but to instead choose a node from $N(v) - \{w\}$ uniformly at random.

\subsection{Theoretical Analysis}
\label{subsec:cnrw-thm}
In this section, we will first introduce the concept of {\em path blocks} and then theoretically show that CNRW and SRW have the same stationary distribution - with the fundamental difference between the two being the how path blocks alternate during the random walk process. Finally, we prove that CNRW produces more accurate samples (i.e. with smaller or equal asymptotic variance than SRW) regardless the graph topology.

\vspace{1mm}
\noindent{\bf Path blocks:}
Intuitively, path blocks divide a random walk's path into consecutive segments. Formally, we have the definition of path blocks:
\begin{definition}\label{def:pb} {\bf [Path blocks].}
Given a random walk over graph $G(V,E)$, denote its path as $X_0, X_1, \dots, X_m, \dots$, where $X_i\in V$. A path block $B_{ij}$ is defined as 
\begin{equation}
B_{ij} = \{X_i, X_{i+1}, \dots, X_{j}\}, j>i.
\end{equation} 
\end{definition}

An interesting observation critical for the design of CNRW is that for path block $B_{ij}$, $X_i$ can be the same as $X_j$, which is called {\em recurrence} in a Markov Chain \cite{Neal2004}. As a typical positive recurrent Markov Chain, SRW will always go back to the same node if the number of steps is sufficiently large. Formally, positive recurrence means that if we denote $M_i = \sum_{n=0}^{\infty}I\{X_n=i|X_0=i\}, i\in V$, then its expected value $E(M_i)$ is always finite. The reason why recurrence is an important concept for CNRW is because, as one can see from its own, the key difference between CNRW and SRW is how we select the next transition when a recurrence happens.

As mentioned in the introduction, for the purpose of CNRW design, one actually has two choices on how to define a recurrence - based on edge-based and node-based path blocks, respectively:
\begin{align}
\text{Edge-based: } &\text{recurrence if } X_i = X_{j-1}, X_{i+1} = X_{j} \\
\text{Node-based: } &\text{recurrence if } X_i = X_{j}
\end{align}
In CNRW (and the following discussions in GNRW), we choose the edge-based design for the following main reason: Note that with the edge-based definition, path blocks separated by recurrences have much longer expected length than with the node-based design, because it takes much more steps for a random walk to travel back to an edge than a node. As such, edge-based path blocks tend to have similar distribution for each block (proved in Theorem~\ref{thm:cnrw_dist}), leading to a smaller inter-pathblock variance, which in turn results in a more significant reduction of asymptotic variance brought by CNRW's without-replacement design (proved in Theorem~\ref{thm:cnrw}). We also conducted extensive experiments to verify the superiority of edge-based design - though the results are not included in this paper due to space limitations.

\noindent{\bf Sampling Distribution:} As discussed in Section~\ref{subsec:randomwalk}, for SRW, the sampling distribution is $\pi_j = k_j/2|E|$. CNRW, on the other hand, is a higher-order Markov Chain. As such, the sampling distribution of CNRW is calculated as in equation (\ref{eq:sta-dist}). We will introduce a specific kind of path blocks first. 

\begin{definition}
\label{def:pb}
{\bf [Path blocks $B(i)$ rooted on edge $e_{uv}$].} Given a random walk with subsequence $u \rightarrow v \rightarrow i \rightarrow u_1 \to u_2 \to \cdots \to u_h \to u \to v \to \cdots$ where $i \neq u$, if there does {\em not} exist $j \in [1, h - 1]$ with $u_i = u$ and $u_{i+1} = v$, then we call the prefix of this subsequence ending on $u_h$ - i.e., $u \rightarrow v \rightarrow i \rightarrow u_1 \to u_2 \to \cdots \to u_h$ an instance of {\em path block} $B(i)$ rooted on $e_{uv}$.
\end{definition} 

What we really mean here is that once a random walk reaches $u \to v$ for the first time, the remaining random walk can be partitioned into consecutive, {\em non-overlapping} path blocks rooted on edge $e_{uv}$, each of which starts with $u \to v$ and ends with the last node visited by the random walk before visiting $u \to v$ again. Using this definition, the following theorem indicates the equivalence of sampling distributions of CNRW and SRW.

\begin{figure}
\centering
\includegraphics[scale=0.3]{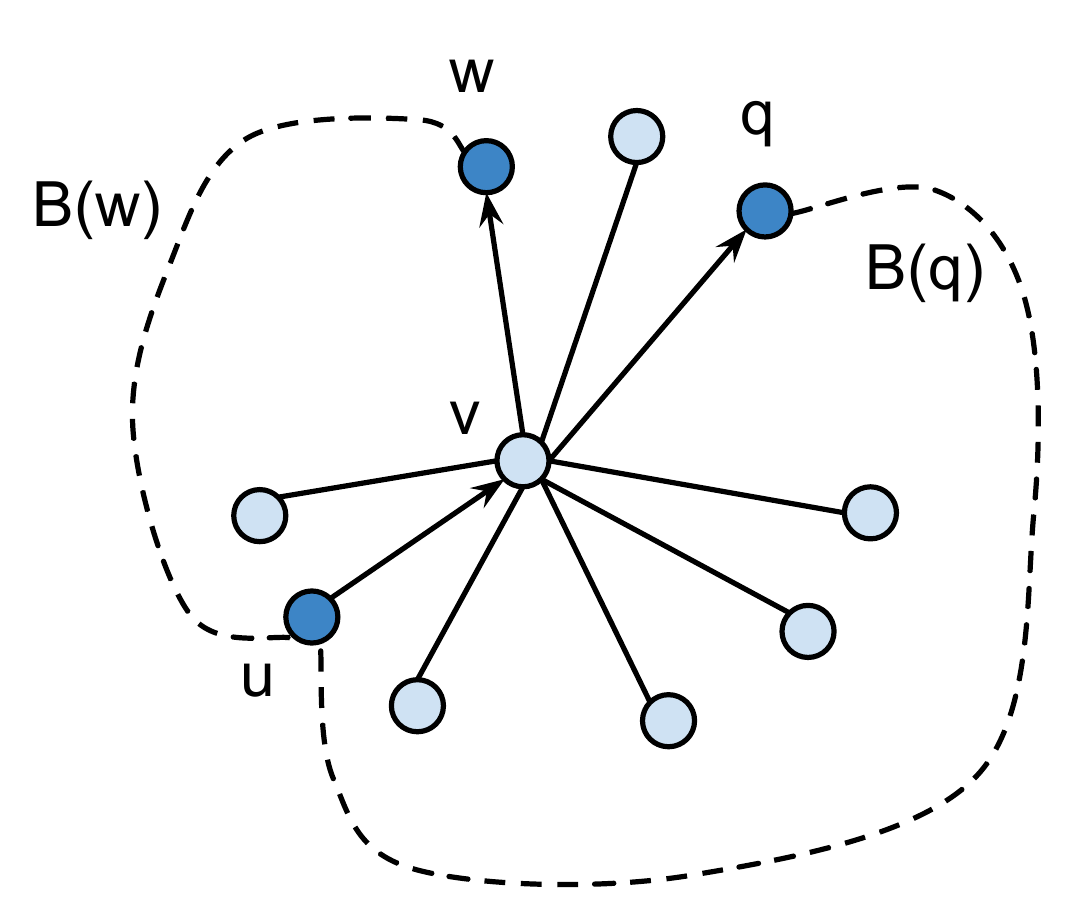}
\caption{Path blocks in CNRW}
\label{fig:cnrw-blocks}
\end{figure}

\begin{theorem}
Given a graph $G(V,E)$, CNRW has the stationary distribution $\pi(v) = k_v/2|E|$.
\label{thm:cnrw_dist}
\end{theorem}

\begin{proof} We will construct the proof in two steps.

\begin{itemize}
    \item First, we prove that the stationary distribution is the same (i.e., $\forall s \in V, \pi(s) = k_s/(2|E|)$ where $k_s$ is the degree of $s$ for simple random walk) after applying CNRW to any {\em single} edge of the graph.
    \item Second, we prove that, $\forall k > 1$, if the stationary distribution is the same after applying CNRW to any set of $(k-1)$ edges, then the stationary distribution will remain the same after we apply CNRW to any additional ($k$-th) edge in the graph.
\end{itemize}

{\bf Step 1.}

We start with the first step - i.e., applying CNRW over only one edge, say $u \to v$. We denote this CNRW as CNRW(1).  Let SRW be the original simple random walk. The most critical concept for our proof is a {\em path block} $B(i)$ for each neighbor of $v$, i.e., $i\in N(v)$. According to the definition in the revised version (also discussed in response to D2), $B(i)$ represents a segment of the random walk that starts with $u \to v \to i$ and ends before the random walk travels back through $u \to v$ next time - e.g., if a segment of random walk is $u \rightarrow v \rightarrow i \rightarrow \cdots \rightarrow x \rightarrow u\rightarrow v \to j \to \cdots \to y \to u \to v$ (the two abbreviated parts in the middle do not contain $u \to v$), then this segment can be rewritten as $B(i) \to B(j) \to u \to v$, where $B(i) = u\rightarrow v \rightarrow i \rightarrow \cdots \rightarrow x$ and $B(j) = u\rightarrow v \rightarrow j \rightarrow \cdots \rightarrow y$ (here $x$ and $y$ can be any arbitrary node that has an edge linked to $u$\footnote{Note that, as further discussed in our response to D2, $B(i)$ and $B(j)$ are not overlapping on $u\rightarrow v$.}).

Let $X^{CNRW(1)}$ and $X^{SRW}$ be a random walk performed by CNRW(1) and SRW, respectively. Let $X^{CNRW(1)}_m$ (resp.~$X^{SRW}_m$) be the node accessed by $X^{CNRW(1)}$ (resp.~$X^{SRW}$) at the $m$-th step. The sampling probability for any node $j \in V$ according to the stationary distribution of the random walks are
\begin{align}
\pi^{CNRW(1)}(j) &= \lim_{n\rightarrow \infty}\frac{\sum_{m=1}^nPr\{X^{CNRW(1)}_m=j\}}{n}, \label{proof:cnrw_1}\\
\pi^{SRW}(j) &= \lim_{n\rightarrow \infty}\frac{\sum_{m=1}^nPr\{X^{SRW}_m=j\}}{n}, \label{proof:srw_1}
\end{align}

We now consider $Y^{CNRW(1)}$ (resp.~$Y^{SRW}$), i.e., an infinite sub-sequence of $X^{CNRW(1)}$  (resp.~$X^{SRW}$) which includes all steps occurring after the first visit of $u \to v$. There are two important observations here:
\begin{itemize}
\item First, both $Y^{CNRW(1)}$ and $Y^{SRW}$ can be completely separated into path blocks $B(i_1)$ $\to$ $B(i_2)$ $\to$ $B(i_3)$ $\to$ $\cdots$ where $\forall h, i_h \in N(v)$. The sole difference between CNRW(1) and SRW is on {\em how a path block is chosen} (i.e., CNRW(1) guarantees $i_1 \neq i_2 \neq \cdots \neq i_{|N(v)|}$ while SRW does not). The internal transitions within a path block, on the other hand, are indistinguishable for both approaches - i.e., for CNRW(1), once $B(i_1)$ is chosen for a slot, the internal transitions within $B(i_1)$ follow the exact same probability distribution as that of a $B(i_1)$ chosen for SRW.
\item Second, since $Y^{CNRW(1)}$  (resp.~$Y^{SRW}$) is equivalent with $X^{CNRW(1)}$ (resp.~$X^{SRW}$) after the first appearance of $u \to v$, as long as $u \to v$ appears in the random walk ($X^{CNRW(1)}$), there must be
\begin{align}
&\lim_{n\rightarrow \infty}\frac{\sum_{m=1}^nPr\{X^{CNRW(1)}_m=j\}}{n}  \\
= & \lim_{n\rightarrow \infty}\frac{\sum_{m=1}^nPr\{Y^{CNRW(1)}_m=j\}}{n} \label{equ:le1}\\
&\lim_{n\rightarrow \infty}\frac{\sum_{m=1}^nPr\{X^{SRW}_m=j\}}{n} \\
= & \lim_{n\rightarrow \infty}\frac{\sum_{m=1}^nPr\{Y^{SRW}_m=j\}}{n} \label{equ:le2}
\end{align}
Note that if $u \to v$ never appears in $X^{CNRW(1)}$, then this Step is already proved: Since CNRW(1) never kicks in, it of course yields the exact same stationary distribution as SRW.
\end{itemize}

Note that, according to the design of CNRW, once we rewrite $Y^{CNRW(1)}$ to path blocks, i.e., $B(i_1) \to B(i_2) \to B(i_3) \to \cdots$, we will see that all $|N(v)|$ path blocks appear in an iterative fashion (e.g., no path block $B(i_x)$ appears twice before another $B(i_y)$ appears once). In other words,
\begin{align}
& \lim_{n\rightarrow \infty}\frac{\sum_{m=1}^nPr\{Y^{CNRW(1)}_m=j\}}{n}\\
= & \frac{\sum_{i \in N(v)} \lim_{n^\prime \to \infty}\frac{\sum_{b=1}^{n^\prime}Pr\{B(i)_b=j\}}{n^\prime}}{|N(v)| \cdot n^\prime}\\
= & \lim_{n\rightarrow \infty}\frac{\sum_{i \in N(v)}\sum_{b=1}^nPr\{B(i)_b=j\}}{|N(v)| \cdot n}. \label{equ:eq1}
\end{align}
where $B(i)_b$ is the node at Step $b$ of path block $B(i)$.

On the other hand, consider the rewrite of $Y^{SRW}$ to path blocks. In this case, the path blocks do not appear in an iterative fashion. Instead, each path block is generated i.i.d.~uniformly from all $|N(v)|$ possible path blocks. Let $P(B(i))$ be the probability for $B_i$ to be chosen at each slot - one can see that $P(B_i) = 1/|N(v)|$. Thus,
\begin{align}
&\lim_{n\rightarrow \infty}\frac{\sum_{m=1}^nPr\{Y^{SRW}_m=j\}}{n} \\
= &\lim_{n\rightarrow \infty} \sum_{i \in N(v)} P(B(i)) \cdot \frac{\sum_{b=1}^nPr\{B(i)_b=j\}}{n}\\
= &\lim_{n\rightarrow \infty}\frac{\sum_{i \in N(v)}\sum_{b=1}^nPr\{B(i)_b=j\}}{|N(v)| \cdot n}. \label{equ:eq2}
\end{align}

One can observe from (\ref{equ:eq1}) and (\ref{equ:eq2}) that
\begin{align}
\lim_{n\rightarrow \infty}\frac{\sum_{m=1}^nPr\{Y^{CNRW(1)}_m=j\}}{n} = \lim_{n\rightarrow \infty}\frac{\sum_{m=1}^nPr\{Y^{SRW}_m=j\}}{n}.
\end{align}
This, combined with (\ref{proof:cnrw_1}), (\ref{proof:srw_1}) and (\ref{equ:le1}), (\ref{equ:le2}), leads to
\begin{align}
\pi^{CNRW(1)}(j) = \pi^{SRW}(j),
\end{align}
i.e., the stationary distribution achieved by CNRW(1) is exactly the same as that of SRW.

{\bf Step 2.}

Now we need to prove the second step: if we assume that stationary distribution is the same after applying CNRW to any set of $(k-1)$ edges, then the stationary distribution will remain the same after we apply CNRW to any additional $k$-th edge in the graph. We denote this CNRW as CNRW($k$), and let $u \to v$ be the additional $k$-th edge. Our definition of path blocks $B(i)$ for $i \in N(v)$ remains exactly the same as in Step 1. Note that
\begin{align}
\pi^{CNRW(k)}(j) &= \lim_{n\rightarrow \infty}\frac{\sum_{m=1}^nPr\{X^{CNRW(k)}_m=j\}}{n} \label{proof:cnrw_k}\\
\pi^{CNRW(k-1)}(j) &= \lim_{n\rightarrow  \infty}\frac{\sum_{m=1}^nPr\{X^{CNRW(k-1)}_m=j\}}{n} \label{proof:cnrw_k-1}
\end{align}

The difference between CNRW(k) and CNRW(k-1) is only on how a path block is chosen. For CNRW(k), once $B(i_1)$ is chosen for a slot, the internal transitions within $B(i_1)$ follow the exact same probability distribution as that of a $B(i_1)$ chosen for CNRW(k-1). Therefore, according to the design of CNRW, once we rewrite CNRW(k) to path blocks, i.e. $B(i_1)B(i_2)B(i_3)\dots$, we will get:

\begin{align}
\lim_{n\rightarrow \infty}\frac{\sum_{m=1}^nPr\{Y^{CNRW(k)}_m=j\}}{n} &= \lim_{n\rightarrow \infty}\frac{\sum_{i \in N(v)}\sum_{b=1}^nPr\{B(i)_b=j\}}{|N(v)| \cdot n}. \label{equ:eq3}
\end{align}

And CNRW(k-1) is just SRW based on the route $u\rightarrow v$ (since edge $e_{u,v}$ is the additional $k$-th edge that has been applied by CNRW(k)), $B(i)$ is chosen uniformly at random in CNRW(k-1). Thus,

\begin{align}
&\lim_{n\rightarrow \infty}\frac{\sum_{m=1}^nPr\{Y^{CNRW(k-1)}_m=j\}}{n} \\
= &\lim_{n\rightarrow \infty} \sum_{i \in N(v)} P(B(i)) \cdot \frac{\sum_{b=1}^nPr\{B(i)_b=j\}}{n}\\
= &\lim_{n\rightarrow \infty}\frac{\sum_{i \in N(v)}\sum_{b=1}^nPr\{B(i)_b=j\}}{|N(v)| \cdot n}. \label{equ:eq4}
\end{align}

We can combine the above equations (\ref{proof:cnrw_k}) - (\ref{equ:eq4}), which leads to

\begin{align}
\pi^{CNRW(k)}(j) = \pi^{CNRW(k-1)}(j).
\end{align}
\end{proof}

One can see from the theorem that, for both SRW and CNRW, the probability for a node to be sampled is proportional to its degree. As such, CNRW can be readily used as a drop-in replacement of SRW.

\vspace{1mm}
\noindent{\bf Improvement on Sampling Efficiency:}  Having proved that CNRW and SRW produce the same (distribution of) samples - we now develop theoretical analysis on the efficiency of CNRW and its comparison with SRW.  Note that, since the sampling distributions of both CNRW and SRW converge asymptotically to the stationary distribution (i.e., probability proportional to a node's degree), for a fair comparison between the efficiency of these two algorithms, we must first establish a measure of the {\em distance} between the sampling distribution achieved at $h$ steps and the (eventual) stationary distribution. The reason is that, only when such a distance measure is defined, we can then compare the number of steps each algorithm requires to reduce the distance below a given threshold $\epsilon$ - apparently, the one which requires fewer steps achieves better sampling efficiency.

For the purpose of this paper, we follow the commonly used {\em asymptotic variance} defined in section~\ref{subsec:measurements}. While there are multiple ways to understand the ``physical meaning'' of this measure, a simple one is that it measures the mean square error (MSE) one would get for an AVG aggregate (on the measure function $f(\cdot)$ used in the definition) by taking into account the nodes encountered in the first $h$ steps of a random walk. One can see that, the lower this MSE is, the closer the sampling distribution at Step $h$ is to the stationary distribution.

With the asymptotic variance measure, we are now ready to compare the sampling efficiency of CNRW and SRW. The following theorem shows that for any measure function $f(\cdot)$ and any value of $h$ (i.e., the number of steps), CNRW always achieves a lower or equal variance - i.e., better or equal sampling efficiency - than SRW. Before presenting the rigid proof, we first briefly discuss the intuition behind this proof as follows.

\begin{figure*}
\centering
\includegraphics[scale=0.48]{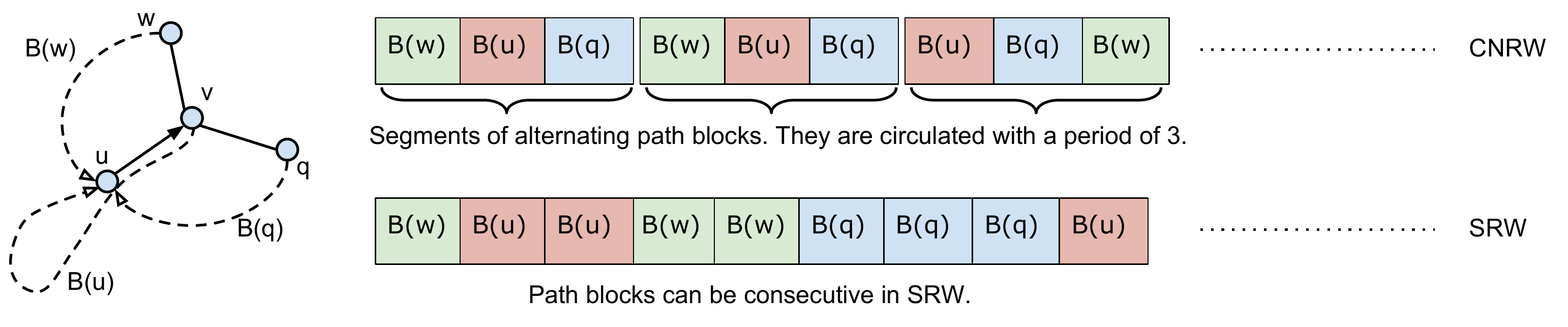}
\caption{Comparison of the block distribution in CNRW and SRW.}
\label{fig:cnrw-ex2}
\end{figure*}

The proof is constructed through induction. As such, we first consider the impact of CNRW on changing the transition after one edge, i.e., as in the above running example, the transition after $u \to v$. A key concept used in the proof is a segmentation of a long random walk into {\em segments} according to $u \to v$. Specifically, every segment of the random walk except the first one starts and ends with $u \to v$. Figure \ref{fig:cnrw-ex2} shows an example of this segmentation. We can see that CNRW are generating segments containing alternating path blocks.

With the segmentation, a key idea of our proof is to introduce an {\em encoding} of each segment according to the first node it visits after $u \to v$. For example, in the case shown in Figure~\ref{fig:cnrw-ex2} where $v$ has three neighbors $\{u, w, q\}$, we have three possible path blocks: $\{B(u), B(w), B(q)\}$. It is important to note that, since we are only considering the change after $u \to v$ at this time, every path block can be considered as being drawn from the exact same distribution no matter if CNRW or SRW is used - because CNRW does not make any changes once the first node after $u \to v$, i.e., the code, is determined, until the next time $u \to v$ is visited. This observation enables us to simple consider CNRW and SRW as sequences of {\em codes} (i.e. we can map path blocks into codes) in the efficiency comparison - as what happens within each path block is anyway oblivious to whether CNRW or SRW is being used.

Now studying the sequence of codes for CNRW and SRW, we make a straightforward yet interesting observation: given a sequence containing $\{B(u), B(w), B(q)\}$ with length $h$, the number of occurrences of $B(u)$, $B(w)$, and $B(q)$ will be the same (at least within $\pm 1$ range) in CNRW, because they will alternately appear every three codes. On the other hand, with SRW, the number of occurrences for $\{B(u), B(w), B(q)\}$ is {\em statistically} equal but have an inherent variance in practice due to randomness. The elimination of this randomness/variance is exactly why CNRW tends to generate samples with a smaller variance than SRW - as rigidly shown in the following theorem.

\begin{theorem}

Given a graph $G(V,E)$, any property function $f$, and the  following two estimators for $ \mu $ based on SRW ($\hat{\mu}$) and CNRW ($\hat{\mu'}$):
\begin{equation}
\hat{\mu}_n = \frac{1}{n}\sum_{t=1}^nf(X_t),\,\,\, \hat{\mu}_n'= \frac{1}{n}\sum_{t=1}^nf(X_t'),
\end{equation}
then CNRW will have no greater asymptotic variance (defined in Section~\ref{subsec:measurements}) than SRW
\begin{equation}
V_{\infty}(\hat{\mu'}) \leq V_{\infty}(\hat{\mu}).
\end{equation}
\label{thm:cnrw}
\end{theorem}

\begin{proof}
The claim of asymptotic variance in Theorem \ref{thm:cnrw} can be divided into 3 steps:
\begin{enumerate}
\item
	We reduce the problem to comparing asymptotic variance when the transition matrix $P$ and $P'$ (for SRW and CNRW accordingly) differ only for transitions involving one node's neighbors. In other words, given nodes $u, v$, we only consider the circulated neighbors rule for CNRW when $u\rightarrow v$, not all the nodes in $G$.
\item
	We then divide the random walk's trace into path blocks. These blocks are all starting with the route of $u\rightarrow v \rightarrow i$, $i\in N(v)$.
\item
	We see that blocks are equally likely, but may have difference distribution for their contents. Blocks in CNRW are alternating, and they will have lower or equal asymptotic variance.
\end{enumerate}

\textbf{Step 1.} Looking at one node's neighbors is enough. If we can prove for only one node's neighbors that CNRW achieves a smaller asymptotic variance, we can recursively apply it to each node and its neighbors. Therefore, given a specific node $v\in V$, one of its neighbor $u$, we only need to consider the circulated neighbors rule to $N(v)$ with an incoming route $u\rightarrow v$.

\textbf{Step 2.} For all the neighbors of $v$: $i\in N(v)$, we have path blocks $B(i)$ rooted on edge $e_{uv}$. For example, $B(w), B(q)$ are two blocks in Figure \ref{fig:cnrw-blocks}. 

\textbf{Step 3.} We will show that the alternating appearance of the blocks $B(i)$ lowers (or at least not increase) the asymptotic variance. We note that the probability of the blocks $P(B(i)) = P(B(j)), i,j\in N(v)$, because $P(u\rightarrow v\rightarrow i) = P(u\rightarrow v\rightarrow j)$. Also, if we denote the number of block $B(i)$ before $M$ blocks as $K_{B(i)}(M)$, then 
\begin{equation}
|K_{B(i)}(M) - K_{B(j)}(M)|\leq 1,\;\;\forall i,j\in N(v), M\geq 1
\end{equation}
holds for CNRW but not SRW. The sampling for the blocks is therefore stratified in CNRW. According to a Lemma in \cite{Neal2004}, the asymptotic variance of CNRW is no more than SRW.

{\em 
    
    \noindent\textbf{Lemma\cite{Neal2004}:}\ \ Let $Z_1,Z_2,\ldots$ be an irreducible
    Markov chain with state space $\{ 0, 1, 2 \}$, whose invariant
    distribution, $\rho$, satisfies $\rho(0)=\rho(1)$.  Let $Q_z$ for 
    $z=0,1,2$ be distributions for
    pairs $(H,L) \in {\mathbb R} \times {\mathbb R}^+$ having
    finite second moments.  Conditional on
    $Z_1,Z_2,\ldots$, let $(H_i,L_i)$ be drawn independently from $Q_{Z_i}$.  
    Define
    \begin{align}
    Z'_i & = & \left\{\begin{array}{ll} 
        Z_i & \mbox{if $Z_i=2$} \\[1pt]
        Z_k + \sum\limits_{j=1}^{i-1} I_{\{0,1\}}(Z_j)\ \ (\mbox{\rm modulo 2}) 
        & \mbox{if $Z_i\ne2$}
    \end{array}\right.
    \end{align}
    where $k=\min\{i:Z_i\ne2\}$.  (In other words, the $Z'_i$ are the same
    as the $Z_i$ except that the positions where $0$ or $1$ occurs have
    their values changed to a sequence of alternating $0$s and $1$s.)
    Conditional on
    $Z_1,Z_2,\ldots$, let $(H'_i,L'_i)$ be drawn independently from $Q_{Z'_i}$.  
    Define two families of estimators as follows:\vspace*{-10pt}
    \begin{equation}
    R_n \,=\, \sum_{i=1}^n H_i \,\Big/\, \sum_{i=1}^n L_i,\ \ \ \ \
    R'_n \,=\, \sum_{i=1}^n H'_i \,\Big/\, \sum_{i=1}^n L'_i
    \end{equation}
    Then the asymptotic variance of $R'$ is no greater than that of $R$.  In other
    words, 
    \begin{equation}
    \lim_{n \rightarrow \infty} nVar(R'_n) \le
    \lim_{n \rightarrow \infty} nVar(R_n)
    \end{equation}
    
}

This lemma justifies the claim that partial stratification of sampling for blocks cannot increase asymptotic variance. And it is easily to extend to the case the state space is $\{0, 1, 2, \dots, k\}$ as long as we keep the stratification of the states. In applying this lemma to the proofs in our theorem, $Z_1, Z_2, \dots$ are identifiers for the type of each block, we can use $0 = B(i_0), 1 = B(i_1), \dots$ etc. Please refer to \cite{Neal2004} for more details. \end{proof}

Theorem~\ref{thm:cnrw} establishes CNRW's superiority while the following Theorem~\ref{thm:barbell} shows how significant the superiority over a concrete example: barbell graph. The probability of propagating from one subgraph to another in CNRW is much greater than the one in SRW.

\begin{theorem}
\label{thm:barbell}
Given a barbell graph $G$, which contains two complete subgraphs $G_1,  G_2$. If we choose the initial node $v\in G_1$, then the expected value of $t_{CNRW}$ and $t_{SRW}$ satisfy that
\begin{equation}
\frac{P_{CNRW}}{P_{SRW}}> \frac{|G_1|}{|G_1|-1}\ln|G_1|
\end{equation}
where $P_{CNRW}$ and $P_{SRW}$ are the probability for CNRW and SRW to travel from $G_1$ to $G_2$.
\end{theorem}

\begin{proof}
We denote the edge $e_{uw}$ ($u\in G_1, w\in G_2$) as the bridging edge in $G$. Now we only need to consider the transition probability of $u$. For SRW, 
\begin{equation}
P_{SRW} = P(u\rightarrow w) = \frac{1}{|N(u)|} = \frac{1}{|G_1|},
\end{equation}
where $|G_1|$ is the number of nodes in $G_1$. For CNRW, 
\begin{align}
P_{CNRW} = P(u\rightarrow w) &= \frac{1}{|G_1| - 1}\sum_{s\in \{G_1 - u\}}P(u\rightarrow w|s \rightarrow u) \label{eq:barbell-1}\\
&= \frac{1}{|G_1|-1}\sum_{i=0}^{|G_1|-1}\frac{1}{|N(u)|-i|} \label{eq:barbell-2}\\
&> \frac{1}{|G_1|-1}\int_{1}^{|N(u)|}\frac{1}{x}dx \\
&= \frac{1}{|G_1|-1}\ln|G_1| \\
&= \left(\frac{|G_1|}{|G_1|-1}\ln|G_1| \right)P_{SRW}
\end{align}
Eq (\ref{eq:barbell-1}) $\Rightarrow$ (\ref{eq:barbell-2}) because the previous accessed neighbors are evenly distributed among $N(u)$.
\end{proof}

\subsection{Algorithm CNRW}
\label{sec:alg_cnrw}
\begin{algorithm}[htb]
    \caption{Circulated Neighbors Random Walk}
    \begin{algorithmic}
    \STATE /* Given $x_0 = u$, $x_1=v$, $b(x_0, x_1) = \emptyset$ */
    \STATE /* Function $b(u,v)$ can be implemented as a HashMap, and their initial value are all $\emptyset$. */	
	        \FOR{$i = 2 \to sample\_size$}
	        	\STATE /* $S$ denotes the next possible candidates */
	        	\STATE $S \gets N(x_{i-1}) - b(x_{i-2}, x_{i-1})$ 
				\IF{$S\neq \emptyset$}
	        		\STATE $x_i \gets$ uniformly choose a node from $S$
	        		\STATE $b(x_{i-2}, x_{i-1}) = b(x_{i-2}, x_{i-1}) \cup \{x_{i}\}$
	        	\ELSE
	        		\STATE $x_i \gets$ uniformly choose a node from $N(x_{i-1})$
	        		\STATE $b(x_{i-2}, x_{i-1}) = \emptyset$
	        	\ENDIF
	        \ENDFOR	        
	\end{algorithmic}
\label{alg:cnrw}
\end{algorithm}

\noindent{\bf Algorithm implementation}. Algorithm \ref{alg:cnrw} depicts the pseudo code for Algorithm CNRW. We note that the only data structure we maintain (beyond what is required for SRW) is a {\em historic hash map of outgoing transitions} $b(u, v)$ for each edge $e_{uv}$ we pass through - i.e., if the random walk has passed through edge $e_{uv}$ before, then $b(u,v)$ contains the neighbors of $v$ which have been chosen previously (during the walk) as the outgoing transitions from $v$.

\noindent{\bf Time and space complexity}. CNRW requires a hash map that continuously records the outgoing transitions for each edge, i.e. the key-value pair: $e_{uv} \rightarrow b(u,v)$. Assume CNRW walks $K$ steps, and the keys $e_{uv}$ are uniformly distributed among all the edges in the graph for CNRW (the same as SRW), then this hash map as a whole occupies expected $O(K)$ space (according to \cite{fredman1984storing} if we are using dynamic perfect hashing). Also, each step's amortized expect time complexity is $O(1)$, so the total expected time complexity for $K$ steps is $O(K)$.

\section{GNRW: GroupBy Neighbors \\Random Walk}
\label{sec:gnrw}

\subsection{Basic Idea}
\label{subsec:gnrw-idea}
Recall from Section~\ref{subsec:cnrw-idea} the main idea of CNRW: Given an incoming transition $u \to v$, we essentially {\em circulate} the next transition among the neighbors of $v$, ensuring that we do not attempt the same neighbor twice before enumerating every neighbor. The key idea of GNRW is a natural extension: Instead of performing the circulation at the granularity of each neighbor (of $v$), we propose to first {\em stratify} the neighbors of $v$ into groups, and then circulate the selection among all groups. In Figure~\ref{fig:gnrw_exs}, for example, if we visit $u \to v$ for the second time, with the last chosen transition being from $v$ to a(nother) node in $S_2$, then this time we will randomly pick one group from $S_1$ and $S_3$ (with probability proportional to the number of not-yet-attempted transitions in each group), and then pick a node from the chosen group uniformly at random. GNRW can be summarized as:

\begin{enumerate}
\item It has a global groupby function $g(\cdot)$ that will partition a node's neighbors into disjoint groups. For example, $g(N(v)) = \{S_1, S_2, \dots, S_m\}$.
\item Each time when the random walk travels from $u$ to $v$, we choose the next candidate group $S_i$ from $N(v)-S(u,v)$ with probability $|S_i|/|N(v)-S(u,v)|$, where $S(u,v)$ is a set of groups we have accessed before.
\item Within $S_i$, we uniformly at random choose the next candidate node $w$ from $S_i - b_{S_i}(u,v)$, where $b_{S_i}(u,v)$ is defined similar to CNRW's - with a range limited in $S_i$.
\item Let $b(u,v) \leftarrow b(u,v)\cup \{w\}$, and $S(u,v) \leftarrow S(u,v)\cup \{S_i\}$. If $b(u,v) = N(v)$, let $b(u,v) = \emptyset$. If $S(u,v)$ contains all the possible groups, let $S(u,v) = \emptyset$.
\end{enumerate}

\begin{figure}
\centering
\includegraphics[width=.4\linewidth]{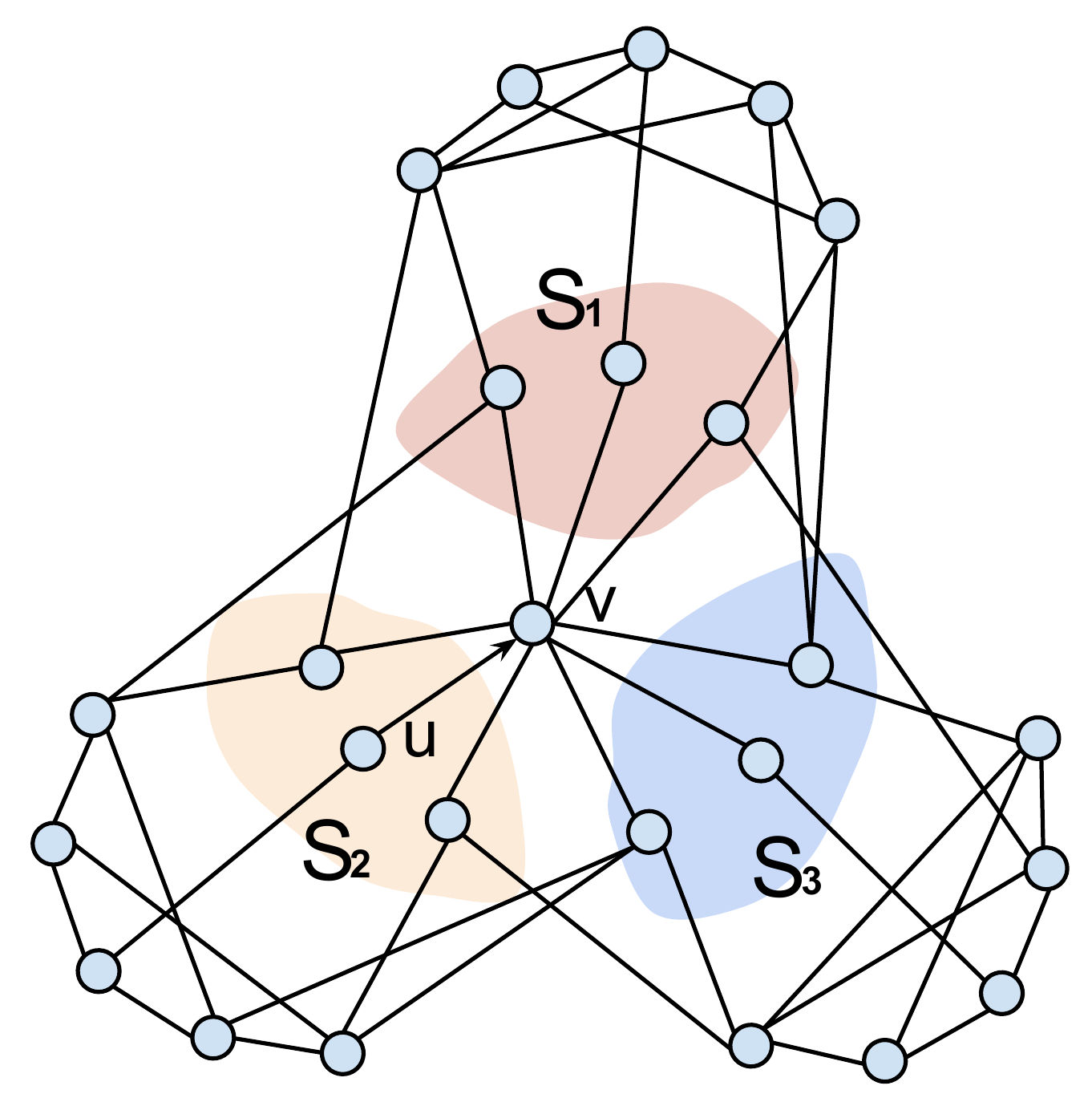}
\caption{An example of partitioning a node's neighbors into 3 groups.}
\label{fig:gnrw_exs}
\end{figure}

One can see from the design of GNRW that, like CNRW, it does not alter the stationary distribution of simple random walk - no matter how the grouping strategy is designed. Also, just like CNRW, GNRW guarantees a smaller or equal asymptotic variance than SRW. This is formalized in the following theorem, the proof of which can be constructed in analogy to Theorem~\ref{thm:gnrw_dist}.

\begin{theorem}
Given a graph $G(V,E)$, GNRW has the stationary distribution $\pi(v) = k_v/2|E|$. And for any property function $f$, the following two estimators for $ \mu $ based on SRW ($\hat{\mu}$) and GNRW ($\hat{\mu^*}$):
\begin{equation}
\hat{\mu}_n = \frac{1}{n}\sum_{t=1}^nf(X_t),\,\,\, \hat{\mu}_n^*= \frac{1}{n}\sum_{t=1}^nf(X_t^*),
\end{equation}
then GNRW will have no greater asymptotic variance than SRW
\begin{equation}
V_{\infty}(\hat{\mu^*}) \leq V_{\infty}(\hat{\mu}).
\end{equation}
\label{thm:gnrw_dist}
\end{theorem}
\begin{proof}
\begin{align}
\pi_j^{GNRW} &= \lim_{n\rightarrow \infty}\frac{1}{n}\sum_{m=1}^nPr\{X^{GNRW}_m=j\} \\
&= \lim_{n\rightarrow \infty}\frac{\sum_{i}\sum_{B(j)\in S_i}\sum_{X_m\in B(j)}Pr\{X_m^{GNRW}=j\}}{n} \label{proof:gnrw-eq2}\\
&= \lim_{n\rightarrow \infty}\frac{\sum_{i}\sum_{X_m\in B(i)}Pr\{X_m^{SRW}=j\}}{n}\label{proof:gnrw-eq3} \\
&= \lim_{n\rightarrow \infty}\frac{1}{n}\sum_{m=1}^nPr\{X^{SRW}_m=j\} \\
&= \pi_j^{SRW}.
\end{align}
Eq. (\ref{proof:gnrw-eq2}) $\Rightarrow$ (\ref{proof:gnrw-eq3}) because GNRW will iterate all the path blocks in $N(v)=\cup_iS_i$, and each path block has the same probability to be accessed ($S_i$ is chosen with probability proportional to its size). The rest of the proof is straightforward and similar to CNRW's.
\end{proof}

\begin{figure*}
\centering
\includegraphics[scale=0.5]{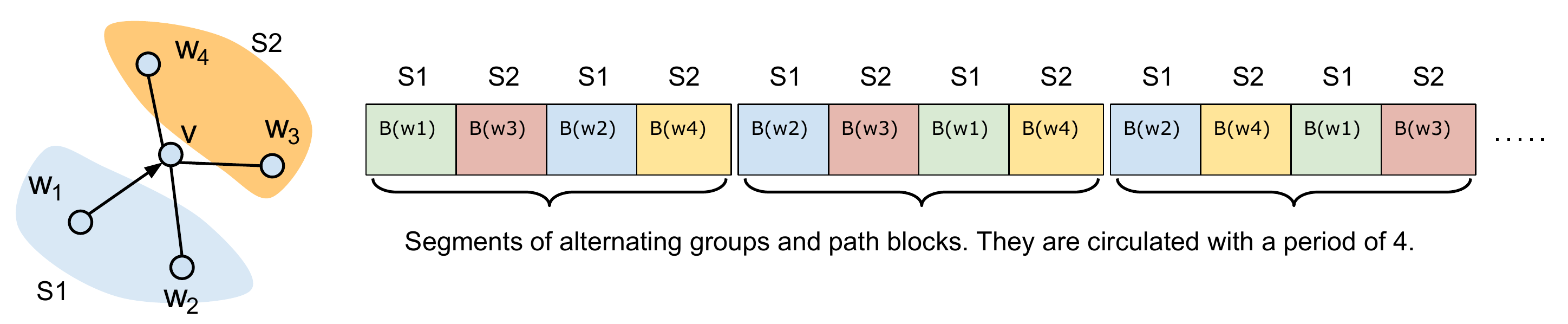}
\caption{A demo of GNRW}
\label{fig:gnrw-ex3}
\end{figure*}

In the following discussions, we first describe the rationale behind GNRW, and then discuss the design of the grouping strategy (for the neighbors of $v$), in other words, the design of groupby function $g(N(v)) = \{S_1, S_2, \dots, S_m\}$. To understand the rationale, we start by considering two extremes of the grouping-strategy design. At one extreme is to group neighbors of $v$ in a completely random fashion. With this strategy, GNRW is exactly reduced to CNRW - i.e., every neighbor of $v$ is circulated, with the order being a random permutation. At the other extreme is the ideal scenario for GNRW - as illustrated in Figure~\ref{fig:gnrw_exs} - when nodes leading to similar path blocks are grouped together. The intuitive reason why GNRW outperforms CNRW in this case is easy to understand from Figure~\ref{fig:gnrw_exs}: Circulating among the three groups, instead of attempting a group more than once, can make the random walk ``propagate'' faster to the entire graph rather than being ``stuck'' in just one cluster.

More formally, the advantage offered by GNRW can be explained according to the path-block encoding scheme introduced in Section~\ref{subsec:cnrw-thm}. Consider an example in Figure~\ref{fig:gnrw-ex3} where $v$ has 4 neighbors.  We encode the path block from each neighbor as random variables $B(w_1), \ldots, B(w_4)$.  Suppose that neighbors are partitioned into two groups: $S_1: \{w_1, w_2\}$ and $S_2: \{w_3, w_4\}$. One can see that GNRW is like stratified sampling, while CNRW is a process of simple random sampling (without replacement). As long as the intra-group variance for the two groups is smaller than the population variance, GNRW offers overall a lower asymptotic variance. To understand why, consider an extreme-case scenario with zero intra-group variance (i.e., $B(w_1) = B(w_2)$, $B(w_3) = B(w_4)$). One can see that, while $GNRW$ achieves zero overall variance, CNRW still has positive variance due to the inter-group variance (i.e., $B(w_1) \neq B(w_3)$).

Having discussed the rationale behind GNRW, we now turn our attention to the design of the grouping strategy. One can see from the above discussion that the main objective here is the group together neighbors that lead to similar path blocks - i.e., random walks starting from nodes in the same group should share similar characteristics. We would like to make two observations for the design: First, {\em locality} is a property widely recognized for social networks - i.e., users with similar attribute values (e.g., age, occupation) tend to have similar friends (and therefore lead to similar path blocks). Thus, grouping neighbors of $v$ based on any attribute is likely to outperform the baseline of random group assignments.

Second, which attribute to use for group assignments also has an implication on the potential usage of samples for analytical purposes. For example, if one knows beforehand that samples taken from a social network will be used to estimate the average age of all users, then designing the grouping strategy based on user age is an ideal design that will likely lead to a more accurate estimation of the average age. To understand why, note that with this grouping strategy, the random walk is likely to quickly ``propagate'' to users of different age groups, instead of being ``stuck'' at a tight community formed by users of similar ages. Thus, if one knows an important aggregate query which the collected samples will be used to estimate, then the grouping strategy should be designed according to the attribute being aggregated in the query. We shall verify this intuition with experimental results in Section~\ref{subsec:eva}.

\subsection{Algorithm GNRW}
\label{sec:alg_gnrw}
\begin{algorithm}[ht]
    \caption{Groupby Neighbors Random Walk}
    \begin{algorithmic}
    \STATE /* Given $x_0 = u$, $x_1=v$, a groupby function $g(\cdot)$ */
    \STATE /* And we assume that all $S(u,v)$ and $b_{S_i}(u,v)$ should be initialized as $\emptyset$. */	
	        \FOR{$i = 2 \to sample\_size$}
	            \STATE $g(N(x_{i-1}) = \{S_1, S_2, \dots, S_m\}$
	            \STATE $CS \gets \{S_1, S_2, \dots, S_m\} - S(u,v)$
				\IF{$CS\neq \emptyset$}
                    \STATE $S_i \gets$ choose a group with probability $|S_i|/|CS|$					
					\STATE $U \gets S_i - b_{S_i}(x_{i-2}, x_{i-1})$ 
					\IF{$U\neq \emptyset$}
		        		\STATE $x_i \gets$ uniformly choose a node from $U$
		        		\STATE $b_{S_i}(x_{i-2}, x_{i-1}) = b_{S_i}(x_{i-2}, x_{i-1}) \cup \{x_{i}\}$
		        	\ELSE
		        		\STATE $x_i \gets$ uniformly choose a node from $S_i$
		        		\STATE $b_{S_i}(x_{i-2}, x_{i-1}) = \emptyset$
		        	\ENDIF
	        	\ELSE
	        	    \STATE $S_i \gets$ uniformly choose a group from $\{S_1, S_2, \dots, S_m\}$
					\STATE $U \gets S_i - b_{S_i}(x_{i-2}, x_{i-1})$ 
					\IF{$U\neq \emptyset$}
		        		\STATE $x_i \gets$ uniformly choose a node from $U$
		        		\STATE $b_{S_i}(x_{i-2}, x_{i-1}) = b_{S_i}(x_{i-2}, x_{i-1}) \cup \{x_{i}\}$
		        	\ELSE
		        		\STATE $x_i \gets$ uniformly choose a node from $S_i$
		        		\STATE $b_{S_i}(x_{i-2}, x_{i-1}) = \emptyset$
		        	\ENDIF
	        	\ENDIF
	        \ENDFOR	        
	\end{algorithmic}
\label{alg:gnrw}
\end{algorithm}

\noindent{\bf Algorithm implementation.} Algorithm \ref{alg:gnrw} depicts the pseudo code for Algorithm GNRW. We note that the data structures we maintain are two {\em hash maps}: $S(u,v)$ and $b_{S_i}(u,v)$. $S(u,v)$ is a mapping from $(u,v)$ to the current set of groups that GNRW have accessed before based on the route $u\rightarrow v$; $b_{S_i}(u,v)$ is a mapping from $(u,v, S_i)$ to the current set of nodes that GNRW have accessed before based on the route $u\rightarrow v$ and the outgoing group $S_i$.

\noindent{\bf Time and space complexity.} GNRW requires two hash maps that continuously records the outgoing groups and edges for each edge, i.e. the key-value pairs: $e_{uv} \rightarrow S(u,v)$ and  $ (e_{uv}, S_i) \rightarrow b_{S_i}(u,v)$. Similar to CNRW, we assume GNRW walks $K$ steps, and it also has the amortized expected $O(K)$ time complexity and $O(K)$ space complexity, because the keys $e_{uv}$ and $(e_{uv}, S_i)$ are uniformly distributed among their possible values \cite{fredman1984storing}.

\section{Discussions}
\label{sec:discussion}
{\noindent\bf CNRW applied to Non-Backtracking Random Walk (NB-SRW)}
It is important to note that the idea of CNRW - i.e., changing
transition upon visiting an edge $u \to v$ from sampling with
replacement to sampling without replacement - is an idea that can be
applied to {\em any} base random walk algorithm, including both SRW
and NB-SRW \cite{58571882}. For example, if we apply the idea to NB-SRW, then the
resulting algorithm (say NB-CNRW) will work as follows: Upon visiting
$u \to v$, instead of sampling the next node with replacement from
$N(v) \backslash u$ (like in NB-SRW), we would sample it without
replacement from $N(v) \backslash u$. Note the difference between
NB-CNRW and the CNRW algorithm presented in the paper (which is based
on SRW): With CNRW, the sampling is done over $N(v)$ while with
NB-CNRW, the sampling is done over $N(v) \backslash u$ - indeed a
carry-over change from NB-SRW.

\vspace{2mm}
{\noindent\bf How does graph size affects CNRW and GNRW.}
First, we note that the graph size is unlikely
to be a main factor in the historic visit probability. To understand
why, consider SRW over an undirected graph and the probability of
going back to the starting node (which, without loss of generality, is
the probability for historic visit to occur). Note that the
probability of going back to the starting node at Step $i$ keeps
decreasing with $i$, until the random walk visits the starting node
again \cite{burioni2005random}. Thus, the historic visit probability
mainly depends on the first $k$ steps after visiting a node, where $k$
is a small constant. In more intuitive terms, a random walk is mostly
likely to go back to a node only a few steps after visiting the node
(because, after only a few steps, the random walk is likely still
within a ``tightly connected'' local neighborhood of the node). This
essentially means that the historic visit probability is unlikely to
be sensitive to graph size - after all, even when the graph size tends
to infinity, the probability of visiting (or re-visiting) a node
within a constant number of steps is unlikely to change much. As an
extreme-case example, growing the graph beyond the $k$-hop
neighborhood of the starting node does not change the historic visit
probability within $k$ steps at all. 

\section{Experiments}
\label{sec:exp}

\subsection{Experimental Setup}

\vspace{1mm}
\noindent {\bf Hardware and platform:} We conducted all experiments on a computer with Intel Core i3 2.27GHz CPU and 64bit Ubuntu Linux OS.

\noindent {\bf Datasets:} We tested three types of datasets in the experiments: well-known public benchmark datasets that are small subsets of real-world social networks, large online social networks such as Google Plus and Yelp, and synthetic graphs (for demonstrating extreme-case scenarios) - e.g., barbell graphs and small clustered graphs. We briefly describe the three types of datasets we used respectively as follows (see the summary of these datasets in Table~\ref{tab:datasets}).

\begin{table*}[ht]
\centering
\begin{tabular}{|r|r|r|r|r|r|}
    \hline  & nodes & edges & average Degree & average clustering coefficient & number of triangles \\ 
    \hline Facebook & 775 & 14006 & 36.14 & 0.47 & 954116 \\ 
    \hline Google Plus & 240276 & 30751120 & 255.96 & 0.51 & 2576826580 \\ 
    \hline Yelp & 119839 & 954116 & 15.92 & 0.12 & 4399166 \\ 
    \hline Youtube & 1134890 & 2987624 & 5.26 & 0.08 & 3056386 \\ 
    \hline Clustering graph & 90 & 1707 & 37.93 & 0.99 & 23780 \\ 
    \hline Barbell graph & 100 & 2451 & 49.02 & 0.99 & 39200 \\ 
    \hline 
\end{tabular}
\caption{Summary of the datasets in the experiments.}
\label{tab:datasets}
\end{table*}

\vspace{1mm}
{\em Public Benchmark:} 

The Facebook dataset is a public benchmark dataset collected from \cite{stanford_dataset}. It is a previously-captured topological snapshot of Facebook and it has been extensively used in the literature (e.g., \cite{Jul2012}). Specifically, the graph we used is from the ``1684.edges'' file. Youtube is another large public benchmark graph collected from  \cite{yang2015defining}. 
In these public benchmark dataset, we simulated a restricted-access web interface precisely according to the definition in Section \ref{subsec:model}, and ran our algorithms over the simulated interface. 

\vspace{1mm}
{\em Large Online Social Graphs:} 

Google Plus\footnote{https://plus.google.com/}. To test the scalability of our algorithms over a large graph, we performed experiments over a large graph we crawled from Google Plus that consists of 240,276 users. We observe that the interface provided by Google Social Graph API strictly adheres to our access model discussed in Section~\ref{subsec:model} - i.e., each API request returns the local neighborhood of one user.

Yelp dataset\footnote{http://www.yelp.com/dataset\_challenge}. We extracted the largest connected subgraph containing 119,839 users (out of 252,898 users) from the dataset. We restored all the dumped JSON data into MongoDB to simulate API requests.

Since we focus on sampling undirected graphs in this paper, for datasets that feature directed graphs, we first converted it to an undirected one by only keeping edges that appear in both directions in the original graph. Note by following this conversion strategy, we guarantee that a random walk over the undirected graph can also be performed over the original directed graph, with an additional step of verifying the existence of the inverse direction (resp.~$v \to u$) before committing to an edge (resp.~$u \to v$) in the random walk.

\vspace{1mm}
{\em Synthetic Graphs:}

We also tested our algorithms over synthetic graphs, such as barbell graphs and graphs with high clustering coefficients for two main purposes: One is to demonstrate the performance of our algorithms over ``ill-formed'' graphs as these synthetic graphs have very small conductance (i.e., highly costly for burning in). The other is to control graph parameters such as number of nodes that we cannot directly control over the above-described real-world graphs. It is important to note that our usage of a theoretical graph generation model does not indicate a belief of the model being a representation of real-world social network topology.


\vspace{1mm}
\noindent{\bf Algorithms:} We implemented and tested five algorithms in the experiments: Simple Random Walk (SRW) \cite{Leskovec2006a}, Metropolis-Hastings Random Walk (MHRW) \cite{hastings1970monte}, Non-Backtracking Simple Random Walk (NB-SRW) \cite{58571882} - a state-of-the-art random walk algorithm which uses an order-2 Markov Chain, and two algorithms proposed in this paper: Circulated Neighbors Random Walk (CNRW) in Section~\ref{sec:cnrw} and Groupby Neighbors Random Walk (GNRW) in Section~\ref{sec:gnrw}. For each algorithm,  we ran it with a query budget ranging from 20 to 1000, and take the returned sample nodes to measure their quality (see performance measures described below). For GNRW, we tested various grouping strategies, as elaborated in Section~\ref{subsec:eva}.

It is important to understand why we included MHRW in the algorithms for testing. Note that, while SRW, NB-SRW and our two algorithms all share the same (target) sampling distribution - i.e., each node is sampled with probability proportional to its degree - MHRW has a different sampling distribution - i.e., the uniform distribution. Thus, it is impossible to compare the samples returned by MHRW with the other algorithms properly. We note that our purpose of including MHRW here is to simply verify what has been recently shown in the literature \cite{Gjoka2010} and \cite{58571882} - i.e., for practical purposes such as aggregate estimation over social networks, the performance of MHRW is much worse than the other SRW based algorithms, justifying the usage of SRW as a baseline in our design.

\vspace{2mm}
\noindent{\bf Performance Measures.} Recall from Section~\ref{subsec:measurements} that a sampling algorithm for online social networks should be measured by {\em query cost} and {\em bias} - i.e., the distance between the actual sampling distribution and the (ideal) target one, which in our case is $\pi(v) = k_v/(2|E|)$. To measure the query cost, one simply counts the number of unique queries issued by the sampler. The measurement of bias, on the other hand, requires us to consider two different methods (and three measures) described as follows.

For a small graph, we measured bias by running the sampler for an extremely long amount of time (long enough so that each node is sampled multiple times). We then estimated the sampling distribution by counting the number of times each node is retrieved, and compared this distribution with the target distribution to derive the bias. Their distances are measured in two forms: (1) KL-divergence \cite{dodge2003oxford}, and (2) $\ell_2$-distance \cite{dodge2003oxford} between the two distribution vectors. Let $P$ and $P_\mathrm{sam}$ be the ideal and measured sampling distribution vectors, respectively.

\begin{itemize}
\item To measure the distance in KL-divergence, we compute \\ $D_\mathrm{KL}(P || P_\mathrm{sam}) + D_\mathrm{KL}(P_\mathrm{sam} || P)$ where
\begin{equation}
D_\mathrm{KL}(P || Q) = \sum_{v_i \in V} \ln\left(\frac{P(v_i)}{Q(v_i)}\right)P(v_i),
\end{equation}
\item For $\ell_2$-norm, we use $\Vert P-P_\mathrm{sam}\Vert _{2}$.
\end{itemize}

Note that compared with the KL-divergence based measure, the $\ell_2$-norm one is more sensitive to ``outliers'' - i.e., large differences on the sampling probability for a single node - hence our usage of both measures in the experiments.

For a large graph like Google Plus and Yelp we used, it is no longer feasible to directly measure the sampling probability distribution (because sampling each node multiple times to obtain a reliable estimation becomes prohibitively expensive). Thus, we considered another measure of bias, {\em aggregate estimation error}, for experiments over large graphs. Specifically, we first used the collected samples to estimate an aggregate over all nodes in the graph - e.g., the average degree or reviews count - and then compare the estimation with the ground truth. One can see that, since SRW, NB-SRW and both of our algorithms all share the exact same target sampling distribution, a sampler with a smaller bias tends to produce an estimation with lower relative error - justifying our usage of it in the experiments.

\subsection{Evaluation}
\label{subsec:eva}

\begin{figure}
\centering
\includegraphics[width=.33\textwidth]{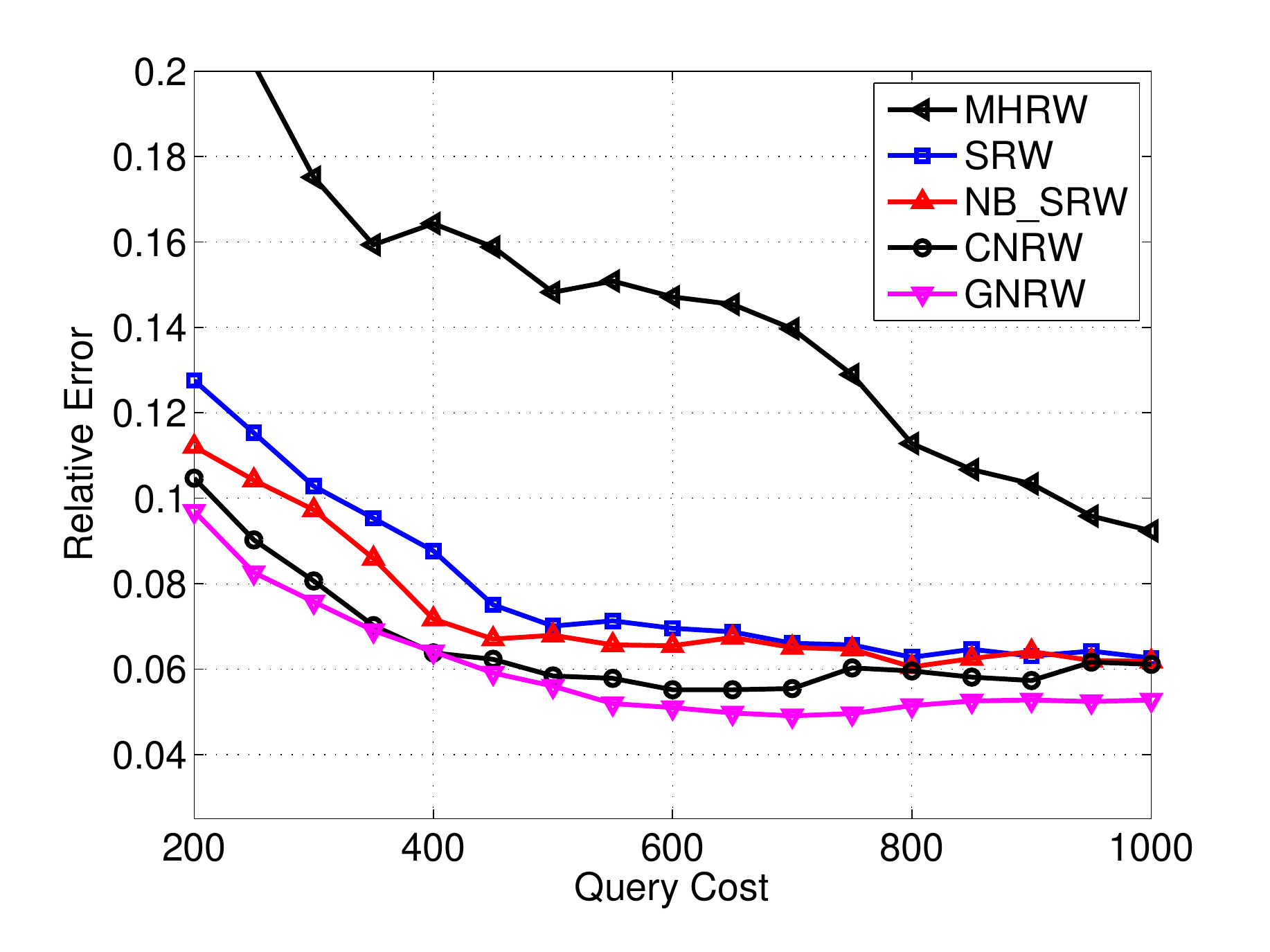}
\caption{Large Google Plus Graph: estimation of average degree.}
\label{fig:gplus}
\end{figure}

\begin{figure*}[htb]
\centering
  \subfigure[Facebook KL-divergence]{
  \includegraphics[width=.23\textwidth]{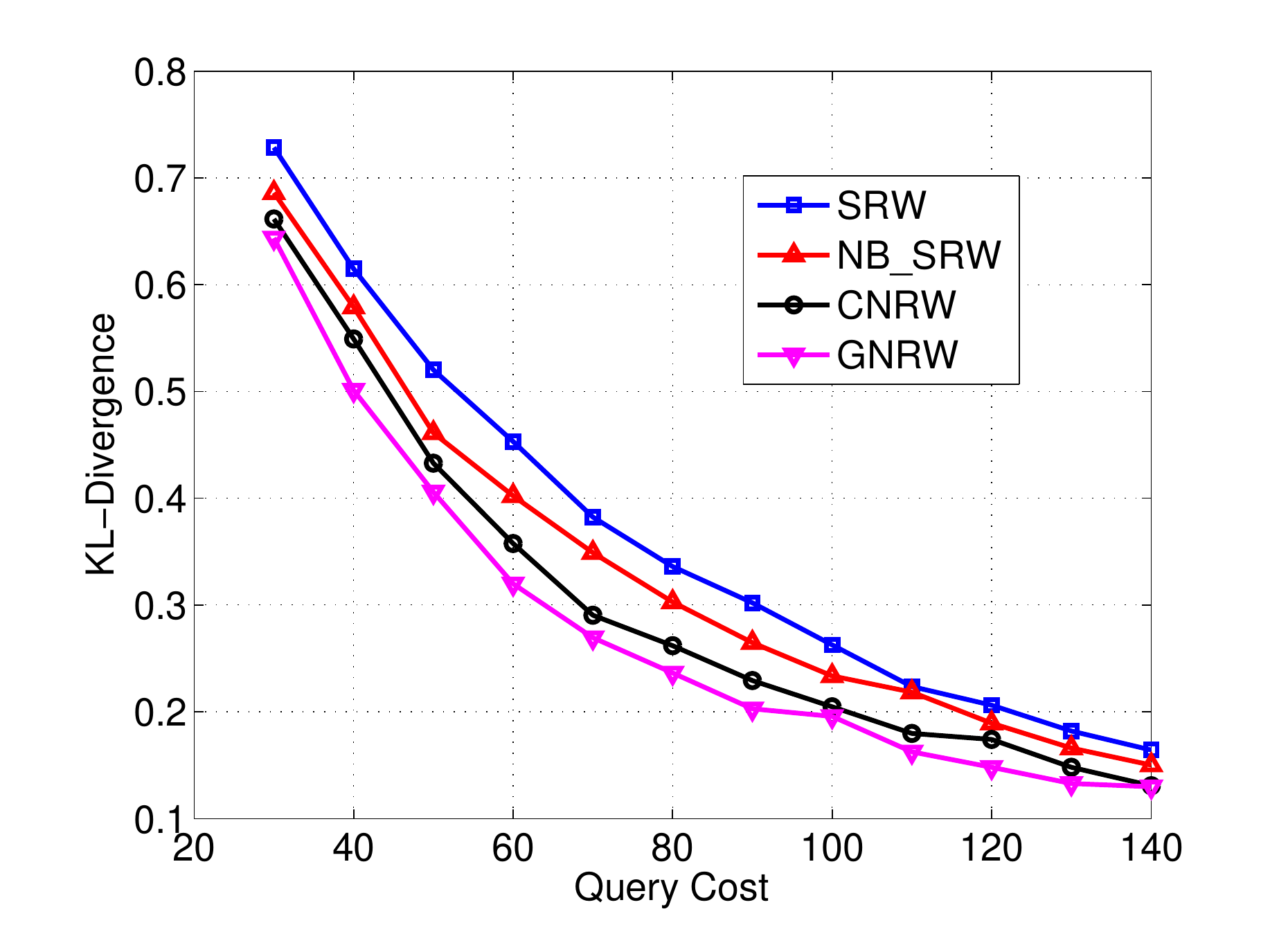}
  \label{fig:facebook2_kl}
  }
    \subfigure[Facebook $\ell_2$-distance]{
  \includegraphics[width=.23\textwidth]{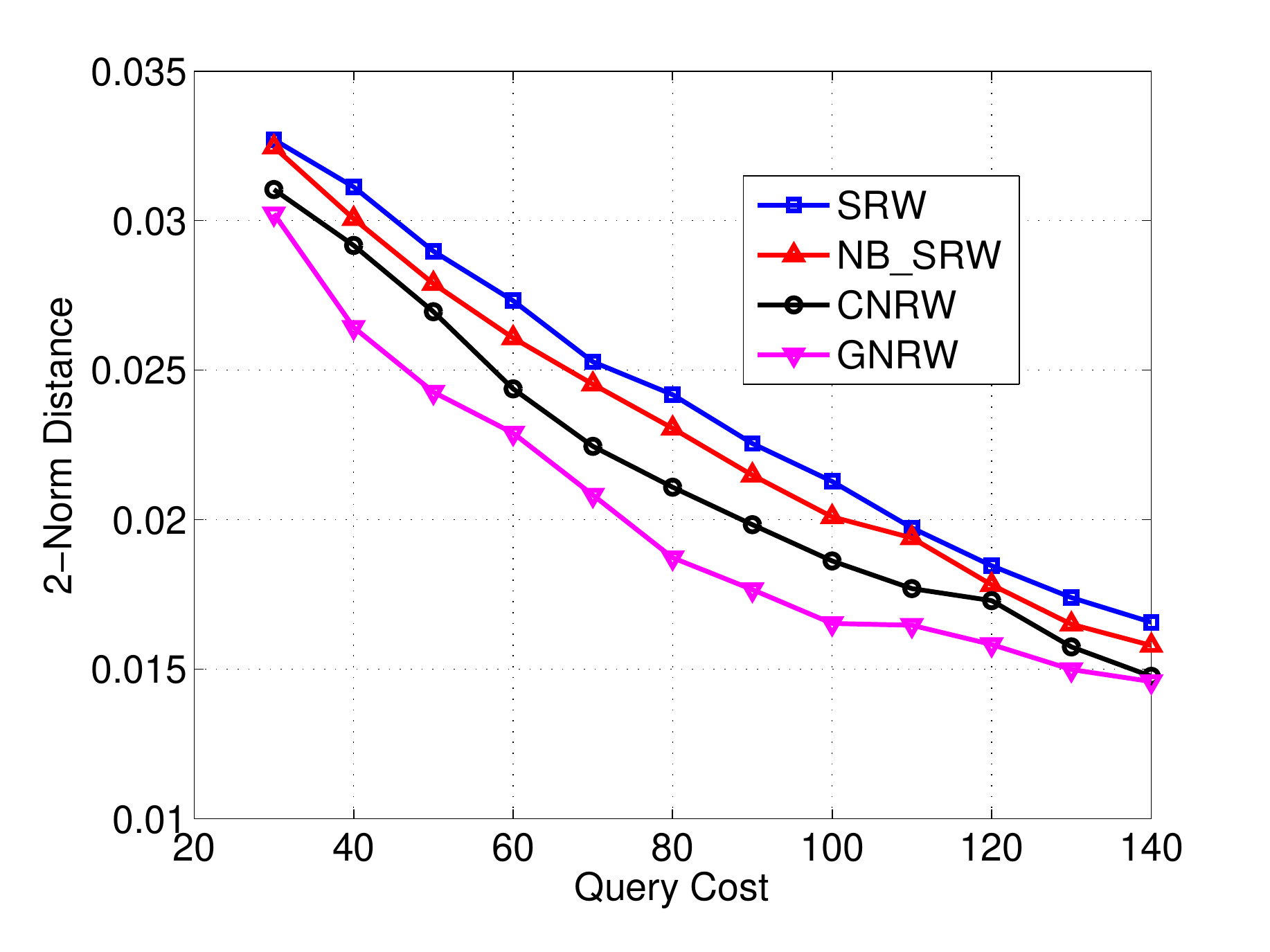}
  \label{fig:facebook2_2norm}
  }
    \subfigure[Facebook Estimation error]{
  \includegraphics[width=.23\textwidth]{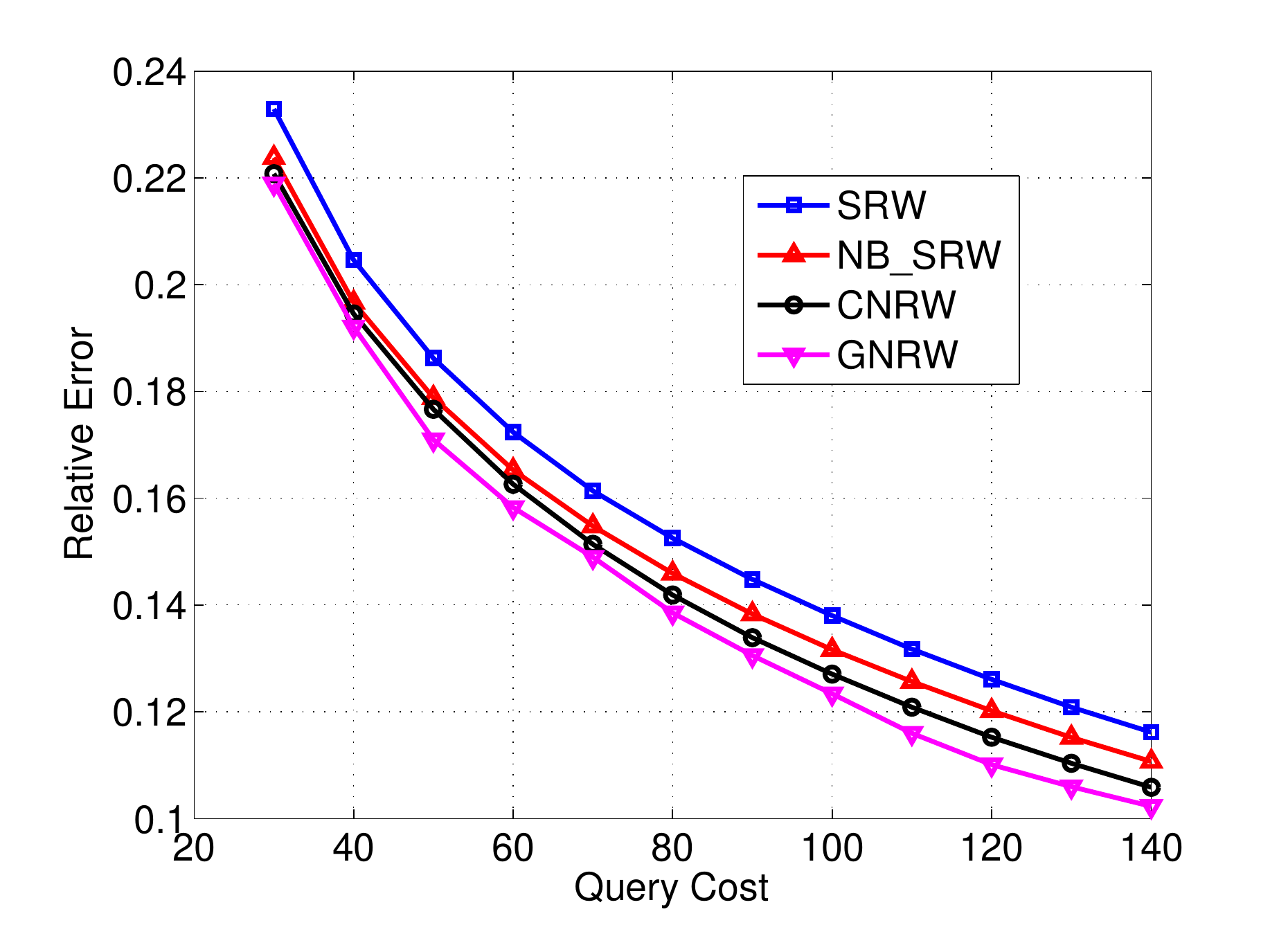}
  \label{fig:facebook2_est}
  }
   \subfigure[Youtube Estimation error]{
    \includegraphics[width=.23\textwidth]{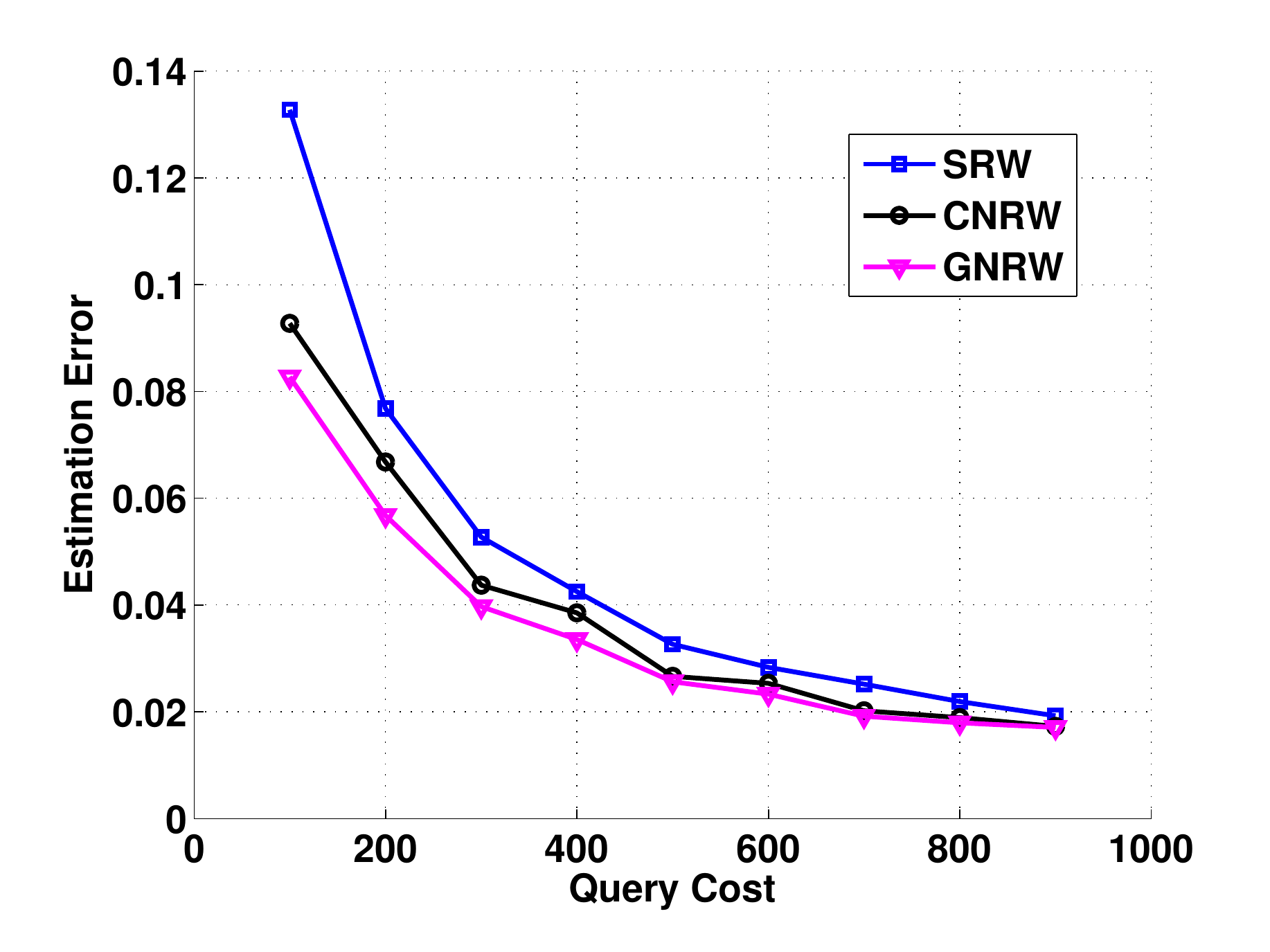}
    \label{fig:youtube_est}
    }
 \caption{Public benchmark datasets}
\label{fig:local-datasets}
\end{figure*}

\begin{figure*}[htb]
\centering

  \subfigure[facebook dataset 1]{
  \includegraphics[width=0.23\linewidth]{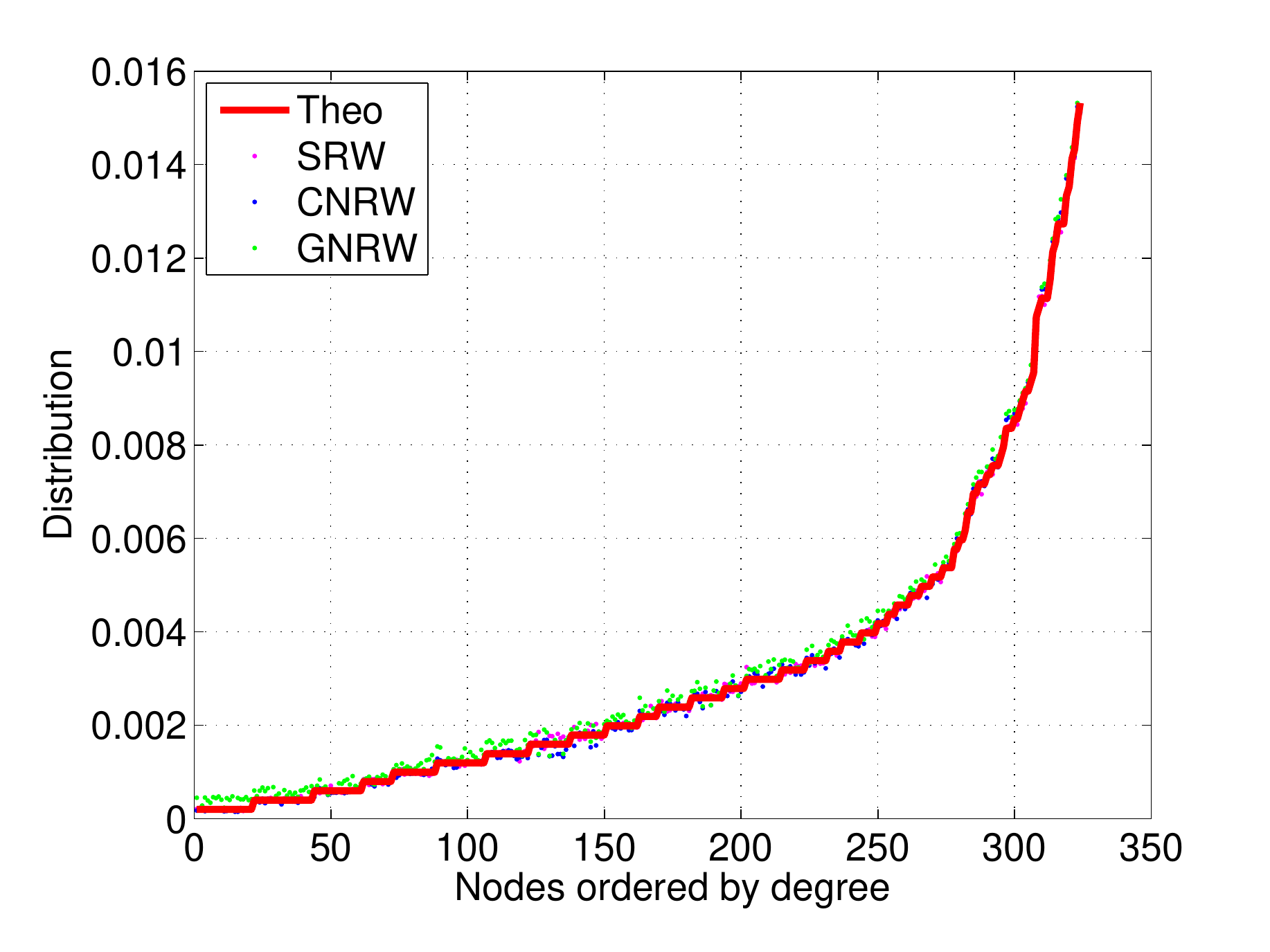}
  \label{fig:distribution-facebook0}
  }
  \subfigure[facebook dataset 2]{
  \includegraphics[width=0.23\linewidth]{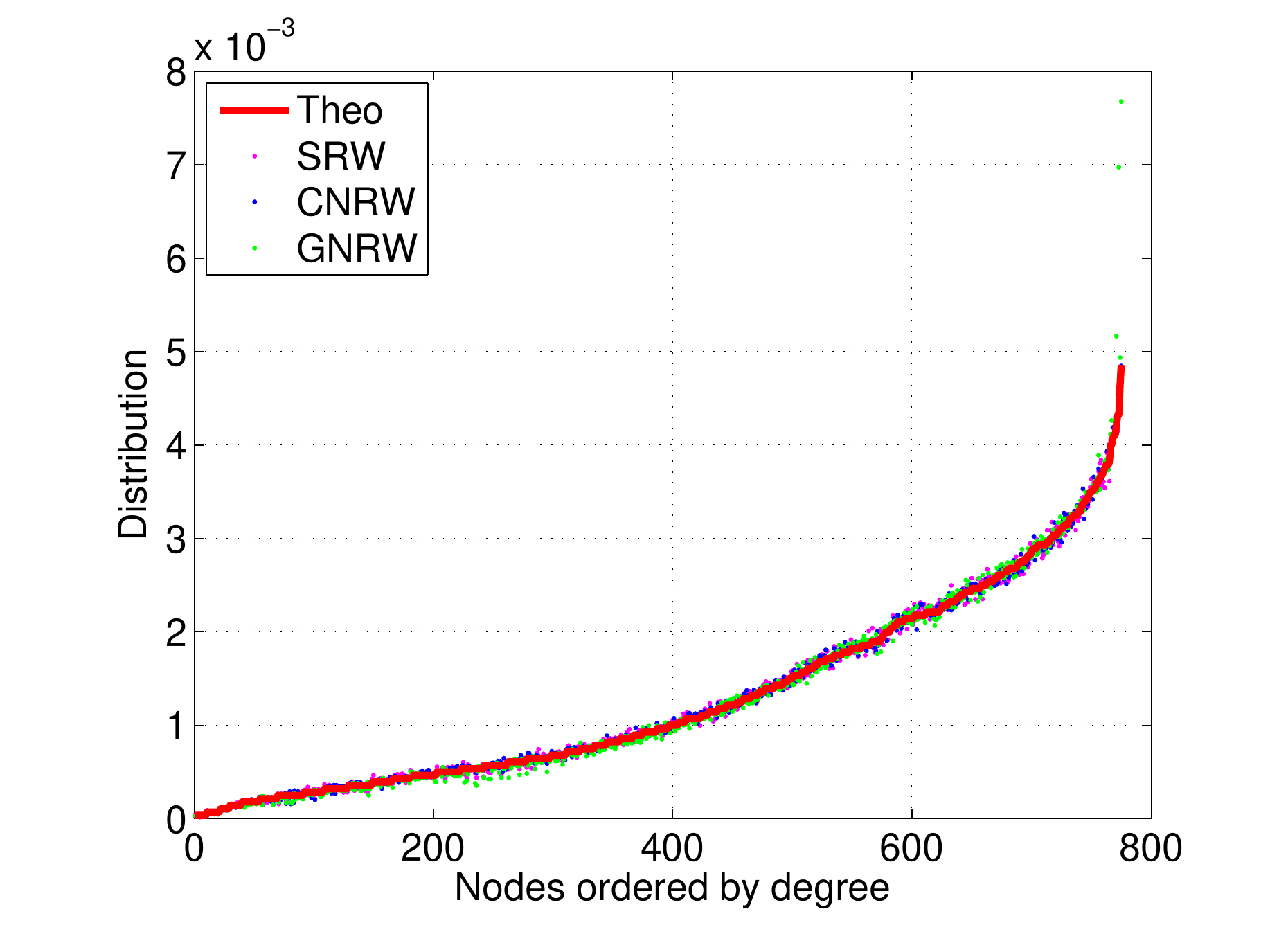}
  \label{fig:distribution-facebook1684}
  }
  \subfigure[facebook dataset 1 (zoomed)]{
  \includegraphics[width=0.23\linewidth]{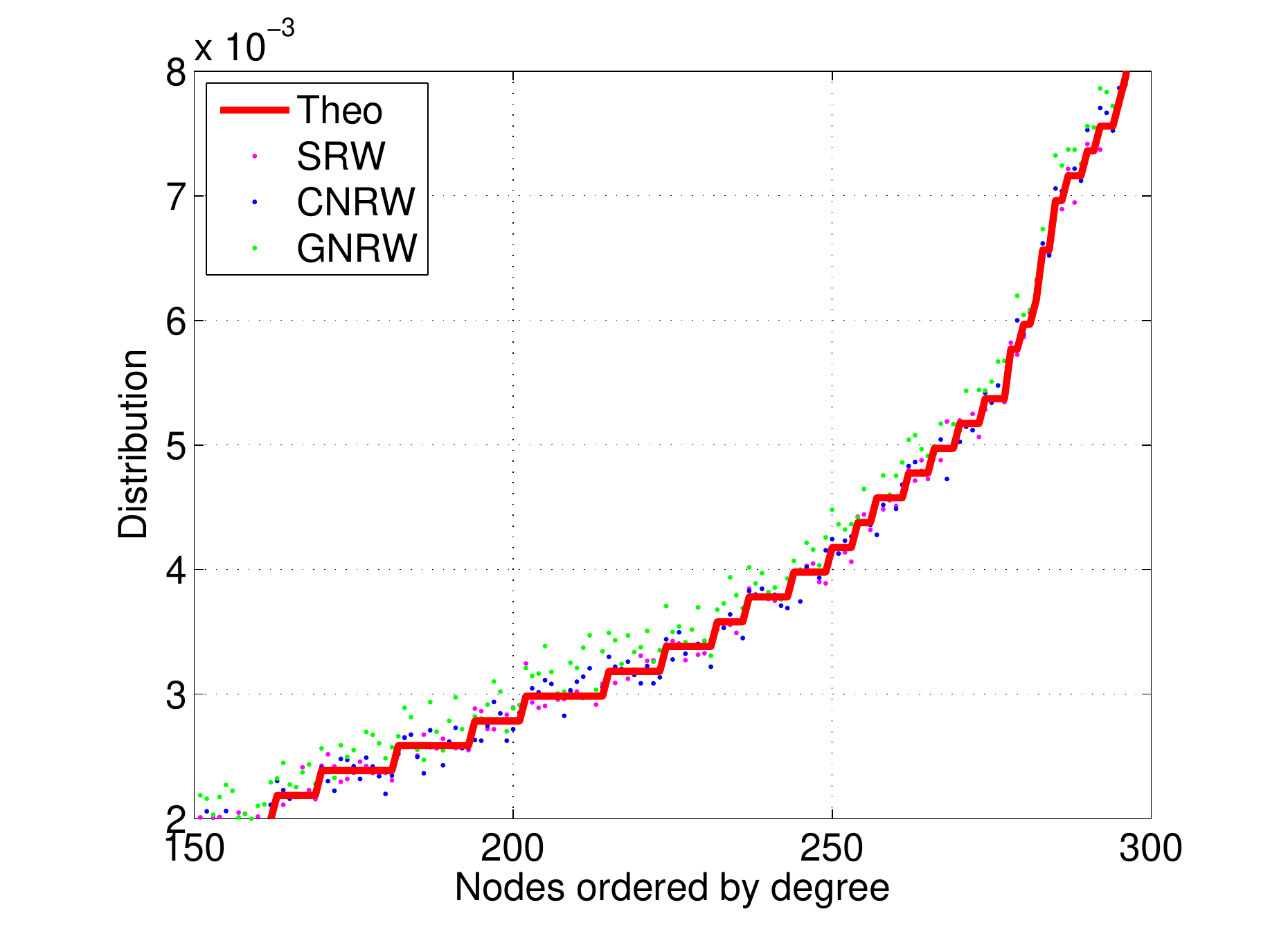}
  \label{fig:distribution-facebook0-zoomed}
  }
  \subfigure[facebook dataset 2 (zoomed)]{
  \includegraphics[width=0.23\linewidth]{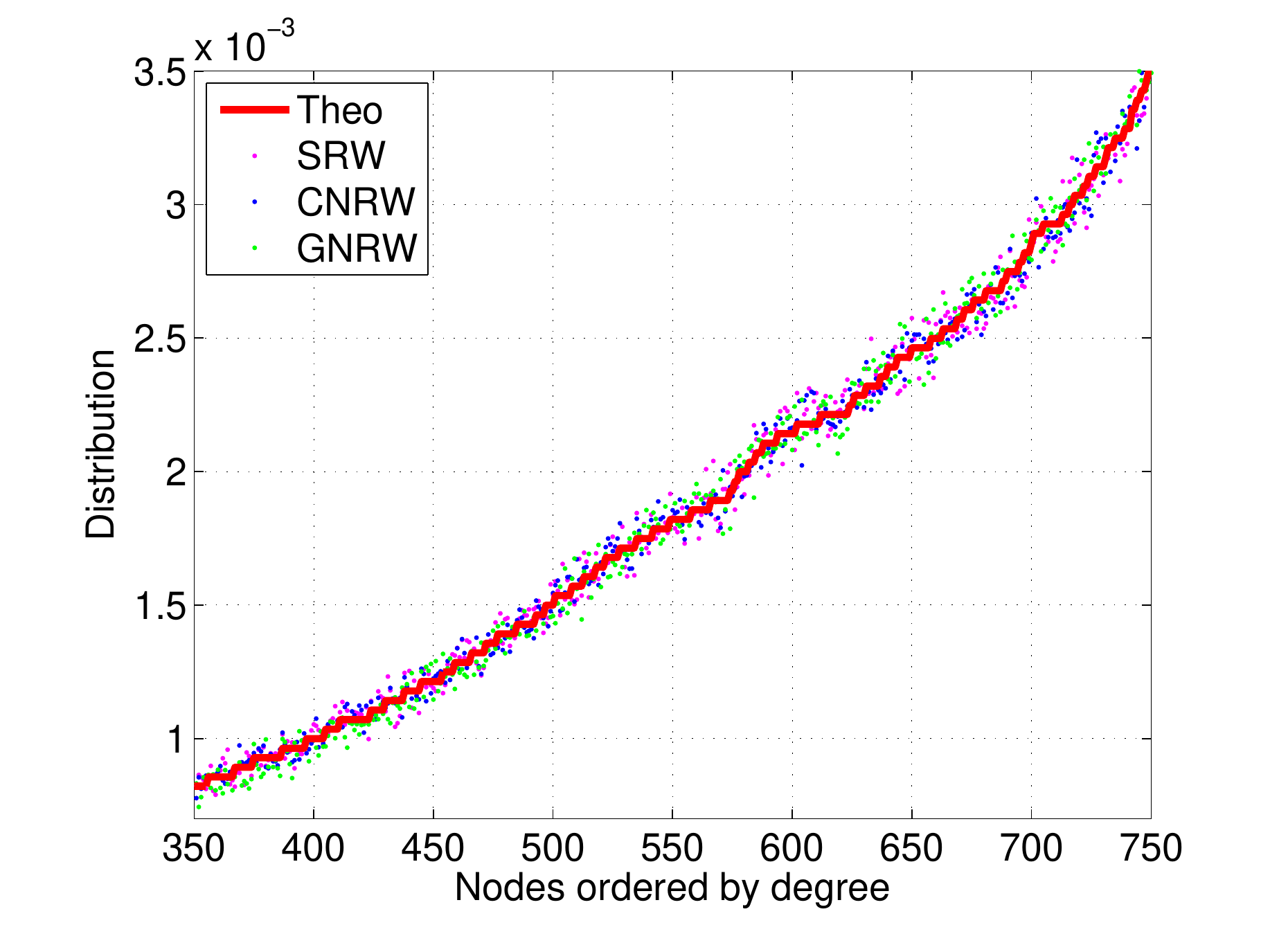}
  \label{fig:distribution-facebook1684-zoomed}
  }
\caption{Distribution of SRW, CNRW and GNRW}
\label{fig:distribution}
\end{figure*}

We start with the experiments that show how SRW, CNRW and GNRW have the same stationary distribution. In Figure~\ref{fig:distribution}, we ran 100 instances of each random walk for 10000 steps in two datasets (both are from the Facebook dataset \cite{stanford_dataset}), and then we used the samples collected from each random walk to calculate the sampling distribution. We ordered the distribution of the nodes by the their degree, and we also included the theoretical distribution of all the nodes (i.e. the red solid line in the figure). One can see that all the three random walks converge to the same stationary distribution.

We then compare of all five algorithms over the Google Plus dataset, with bias measure being the relative error for estimating the average degree. Figure~\ref{fig:gplus} depicts the change of relative error with query cost. One can make two observations from the figure. One, our proposed algorithms CNRW and GNRW significantly outperform the other algorithms. For example, to achieve a relative error of 0.06, CNRW and GNRW only requires a query cost of around 486 and 447, respectively, while SRW requires over 800, NB-SRW requires 795, and MHRW never achieves a relative error under 0.08 after issuing 1000 queries. Second, we can also observe from the figure that MHRW performs much poorer than the other algorithms. Thus, we do not further include MHRW in experimental results in the paper.

With Figure~\ref{fig:gplus} establishing the superiority of our algorithms over the existing ones with the relative error measure, we also confirmed the superiority with the other two measures, KL-divergence and $\ell_2$-norm, this time over the public benchmark dataset. Figure~\ref{fig:local-datasets} depicts the results for Facebook and Youtube. One can observe from Figure~\ref{fig:local-datasets} that our CNRW and GNRW algorithms consistently outperform both SRW and SRW according to all three measures being tested. In addition, GNRW outperforms CNRW, also for all three measures.

To further study the design of GNRW, specifically the criteria for grouping nodes together, we tested GNRW with three different grouping strategies: random grouping (i.e., GNRW-By-MD5, as we group nodes together according to the MD5 of their IDs), grouping by similar degrees (GNRW-By-Degree), and grouping by the value of an attribute ``reviews count'' (GNRW-By-ReviewsCount). Figure~\ref{fig:gnrws-grp} depicts the performance of all three strategies. One can make an interesting observation from the figure: While all three variations significantly outperform the baseline SRW algorithm, the best-performing variation indeed differs when the relative error  is computed from different aggregates. Specifically, when the aggregate is average degree, GNRW-By-Degree performs the best. When the aggregate is average review count, on the other hand, GNRW-By-ReviewsCount performs the best. This verifies our discussions in Section~\ref{subsec:gnrw-idea} - i.e., if the aggregate of interest is known before hand, choosing the grouping strategy in alignment with the aggregate of interest can lead to more accurate aggregate estimations from samples.

\begin{figure}
\centering
  \subfigure[Estimate Average Degree]{
  \includegraphics[width=.22\textwidth]{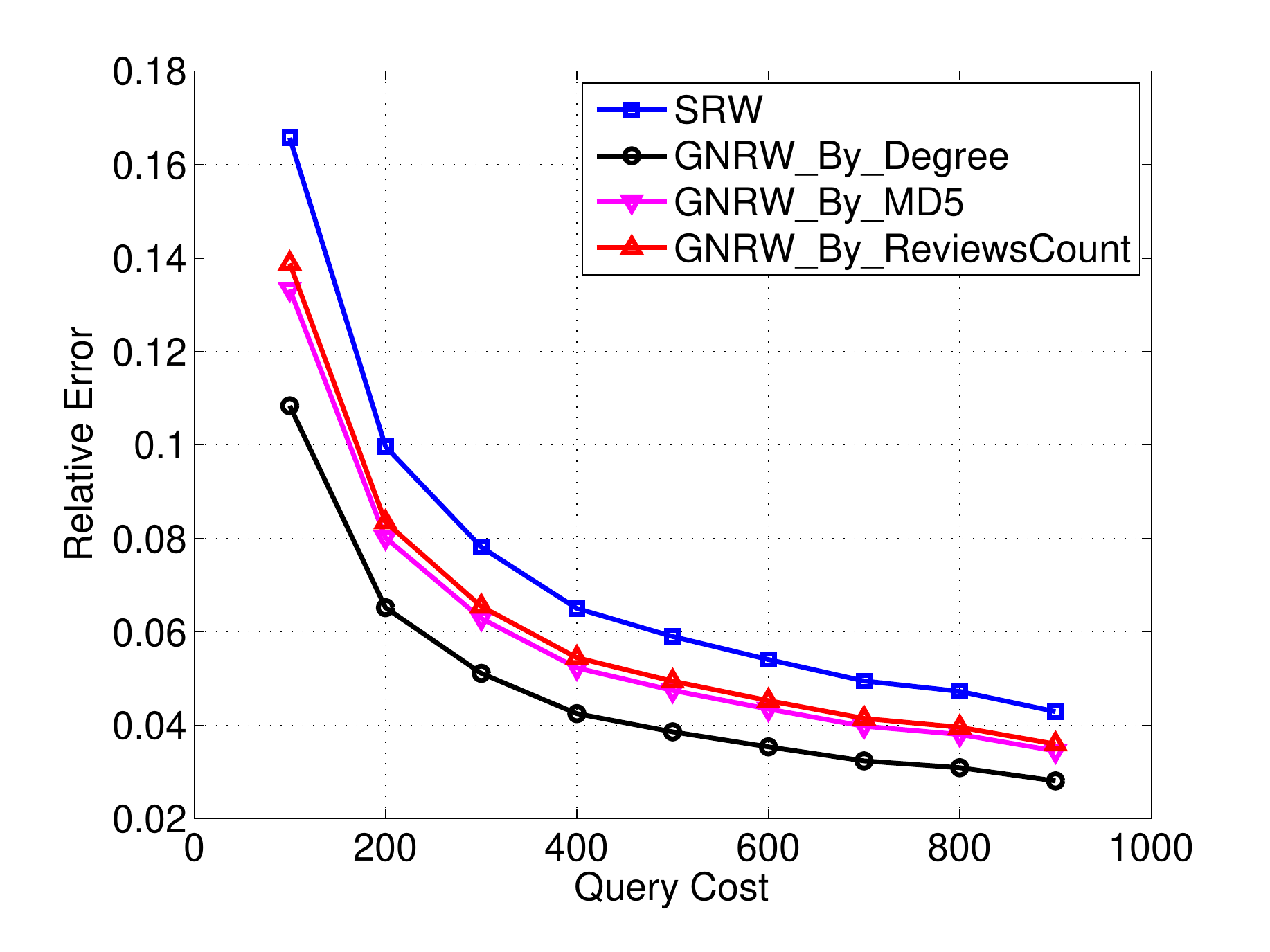}
  \label{fig:yelp-degree}
  }
  \subfigure[Estimate Average Reviews Count]{
  \includegraphics[width=.22\textwidth]{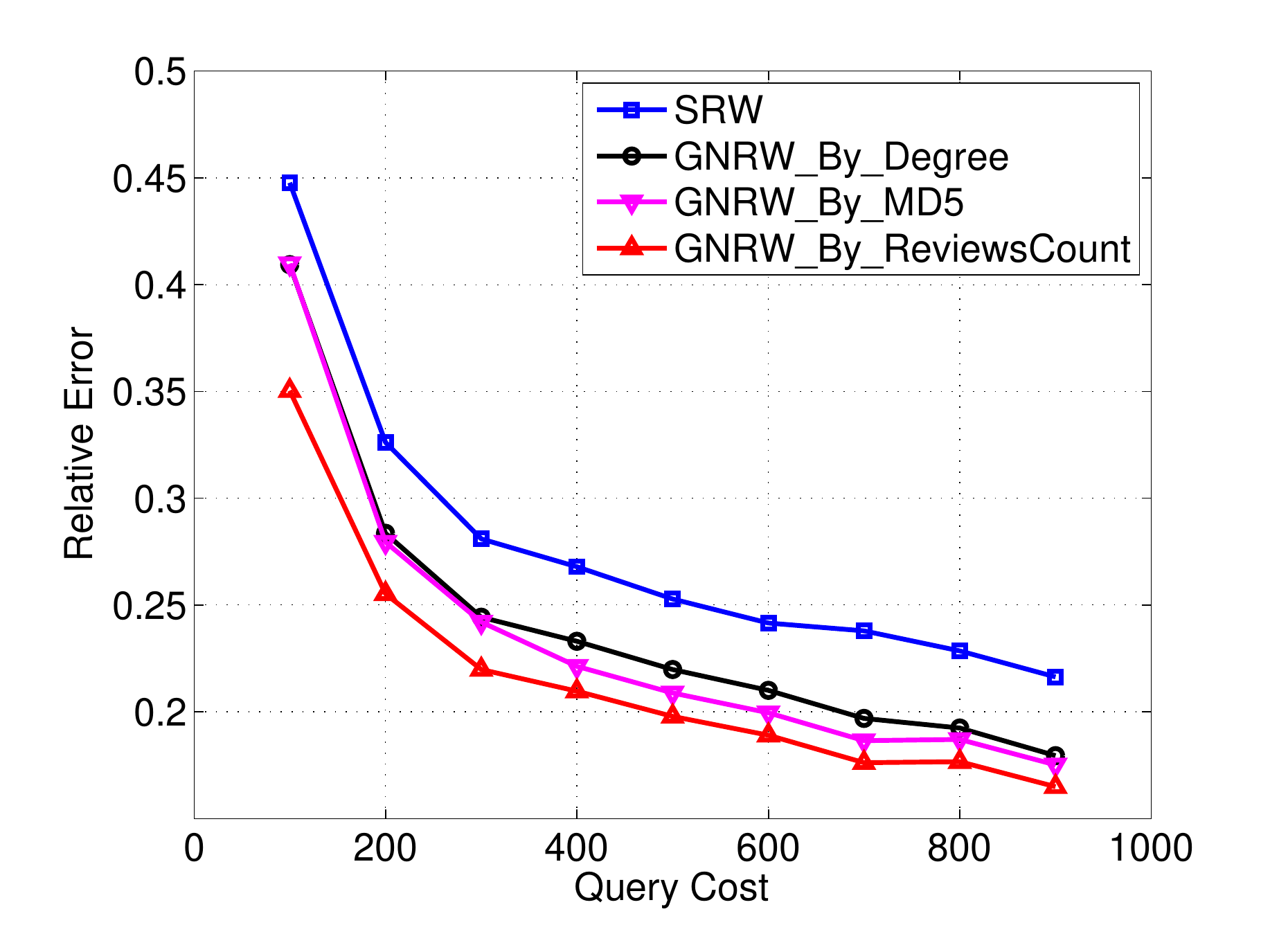}
  \label{fig:yelp-reviewCounts}
  }
\caption{Yelp dataset: GNRW strategies.}
\label{fig:gnrws-grp}
\end{figure}

Finally, we studied the performance of our algorithms over two ``ill formed'' graphs, the clustered graph and a barbell graph. The clustered graph is combined with 3 complete graphs with graph sizes as 10, 30 and 50. We also vary the size of the barbell graph from 20 to 56 nodes, to observe the change of performance of the algorithms with graph size.  The results are shown in Figures~\ref{fig:small-cluster} and \ref{fig:barbell-size}. One can see from the results that, even for these ill-formed graphs, our proposed algorithms consistently outperform SRW and NB-SRW for varying graph sizes, according to all three bias measures being used.

\begin{figure*}
\centering

  \subfigure[KL-divergence]{
  \includegraphics[width=.31\textwidth]{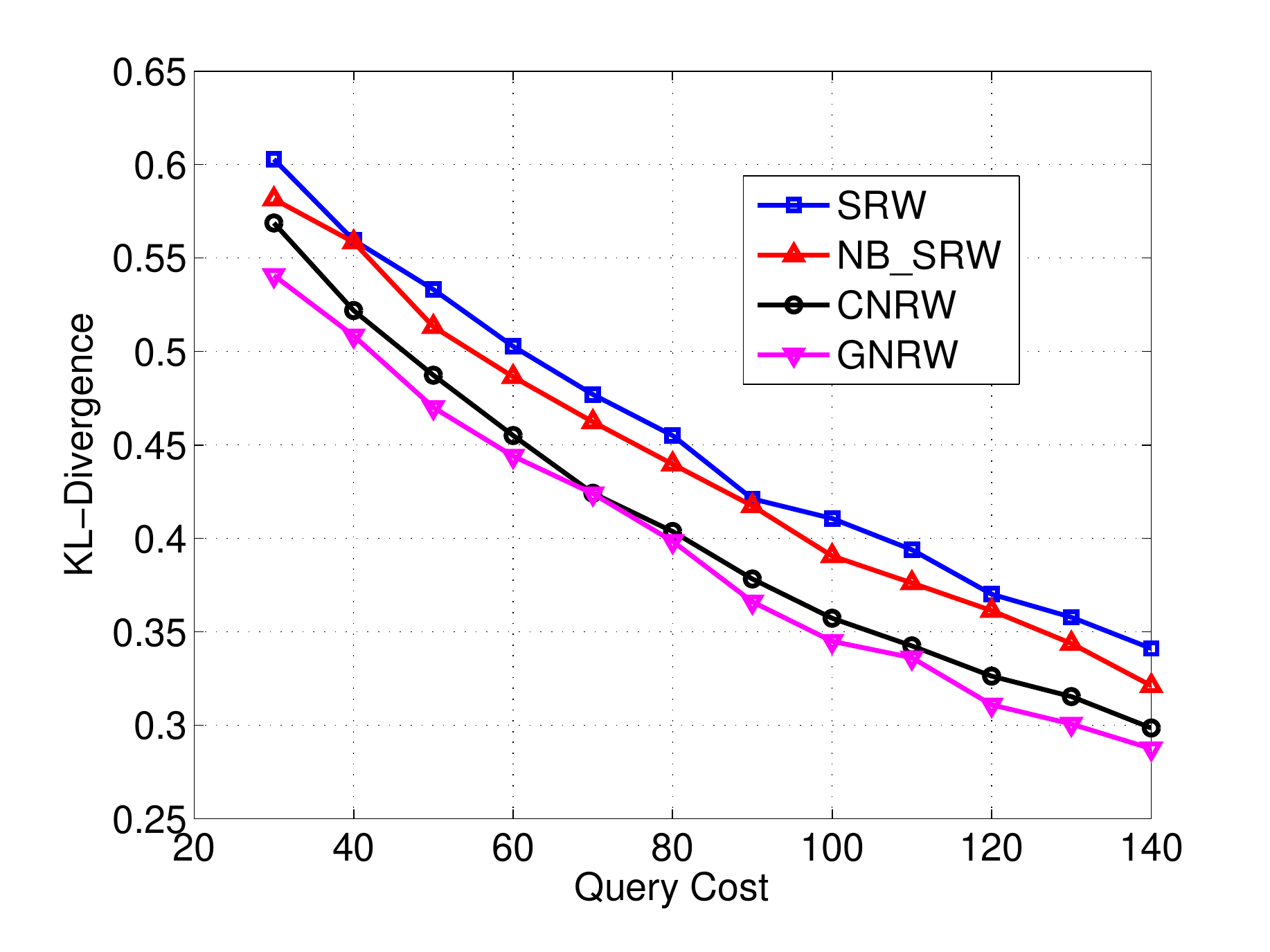}
  \label{fig:cluster_kl}
  }
  \subfigure[$\ell_2$-distance]{
  \includegraphics[width=.31\textwidth]{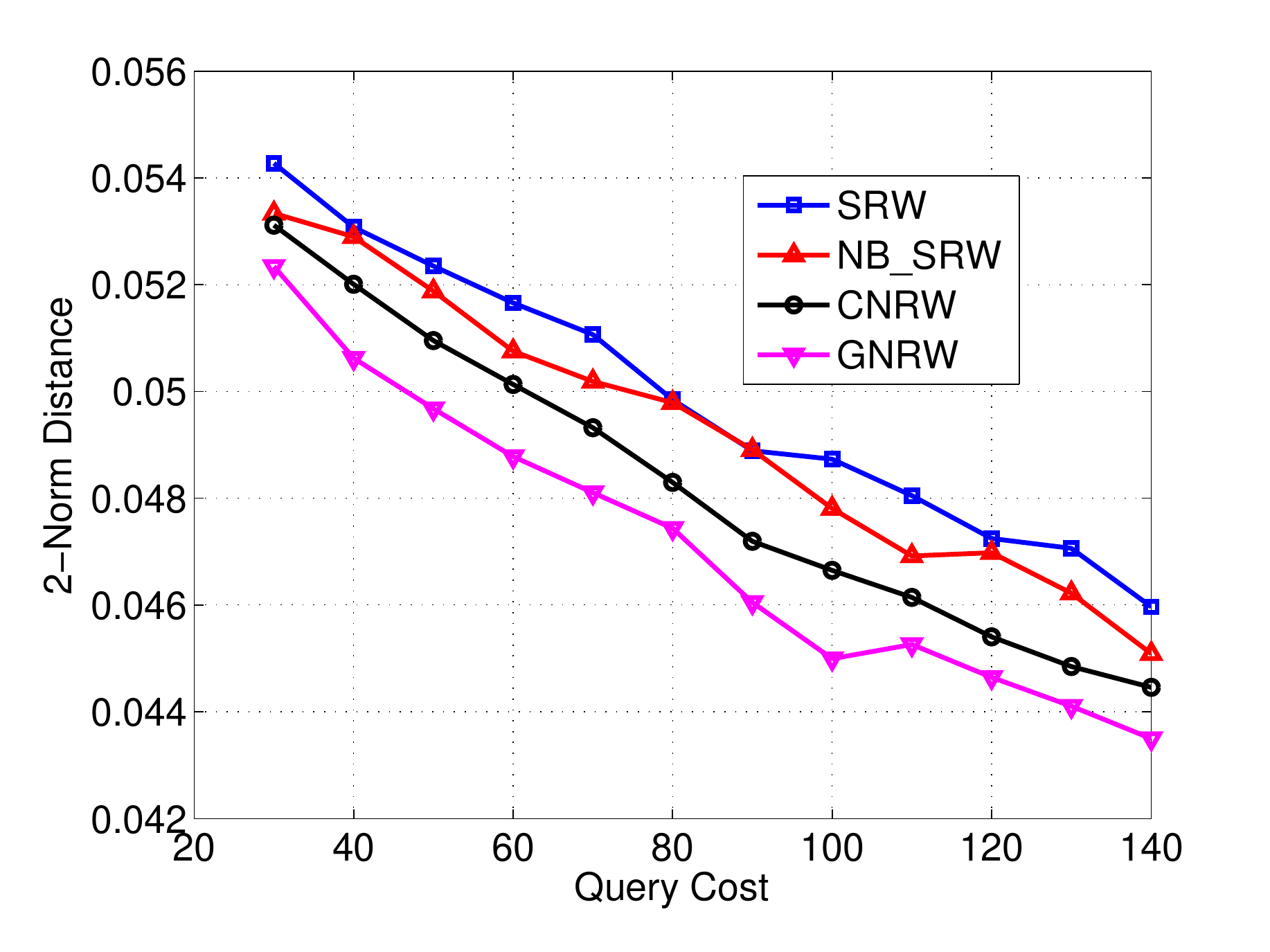}
  \label{fig:cluster_2norm}
  }
  \subfigure[Estimation error]{
  \includegraphics[width=.31\textwidth]{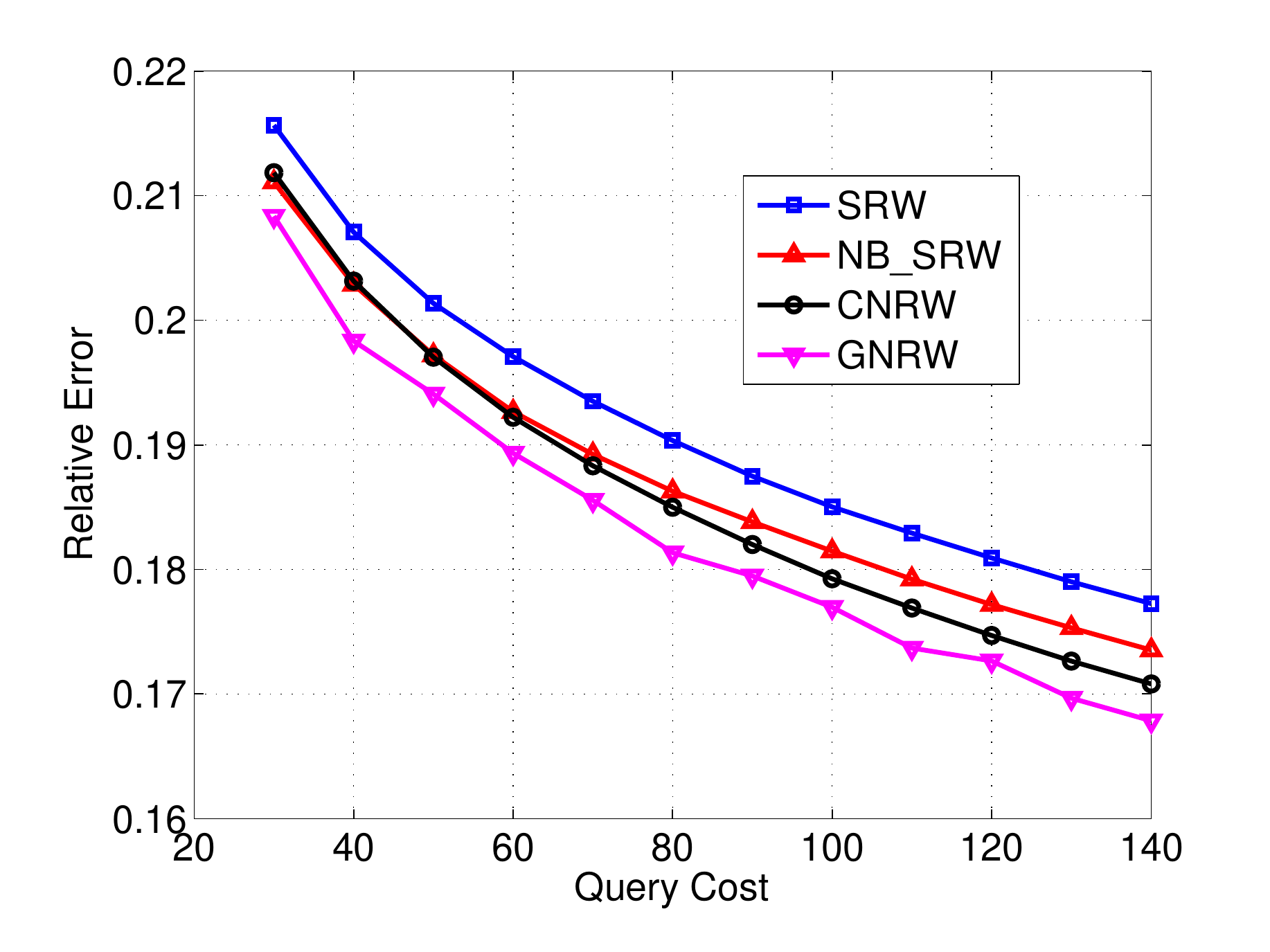}
  \label{fig:cluster_est}
  }
\caption{Synthetic datasets: clustered graph}
\label{fig:small-cluster}
\end{figure*}

\begin{figure*}
\centering
  \subfigure[KL-divergence]{
  \includegraphics[width=.31\textwidth]{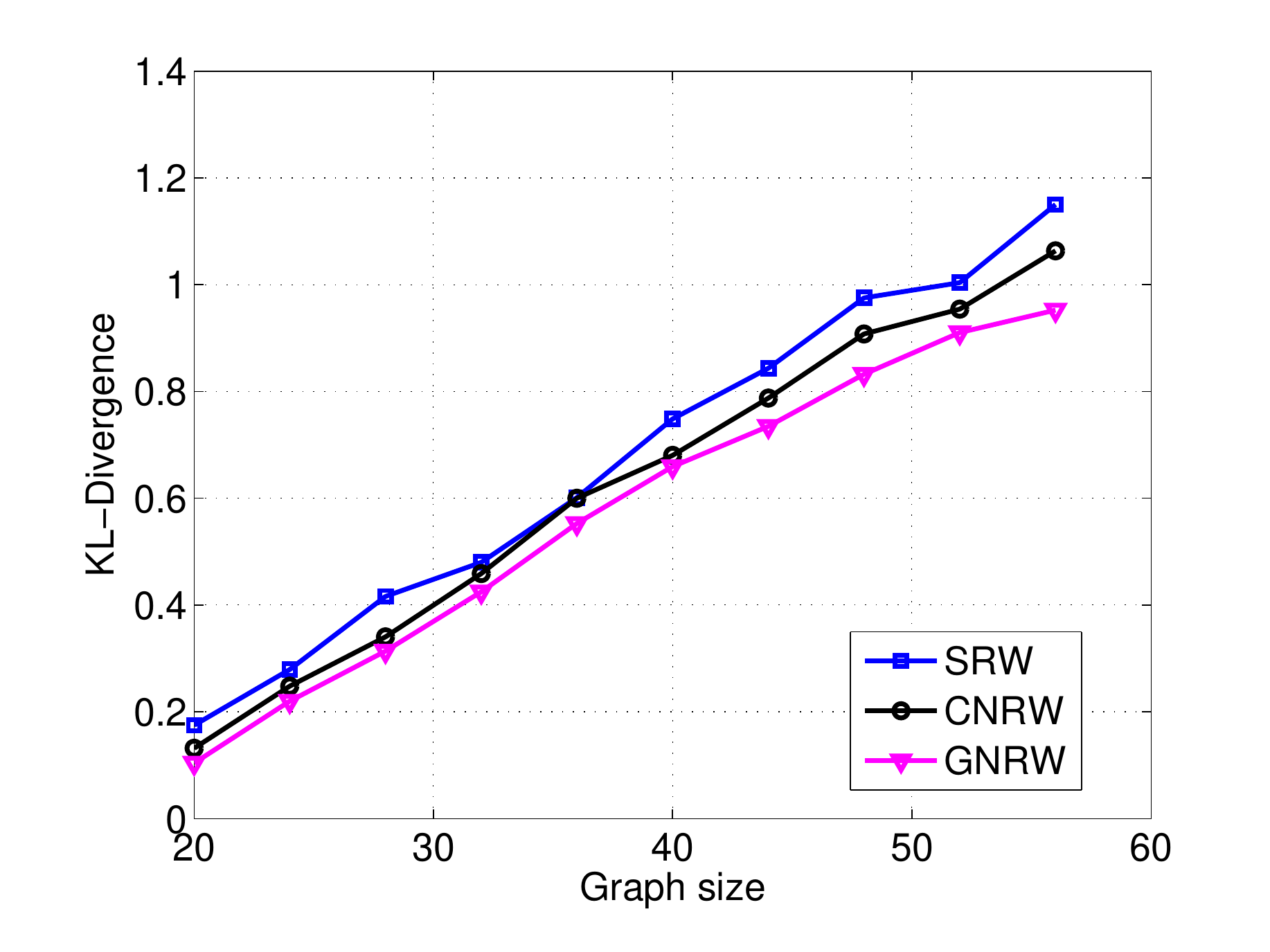}
  \label{fig:barbell_size_kl}
  }
  \subfigure[$\ell_2$-distance]{
  \includegraphics[width=.31\textwidth]{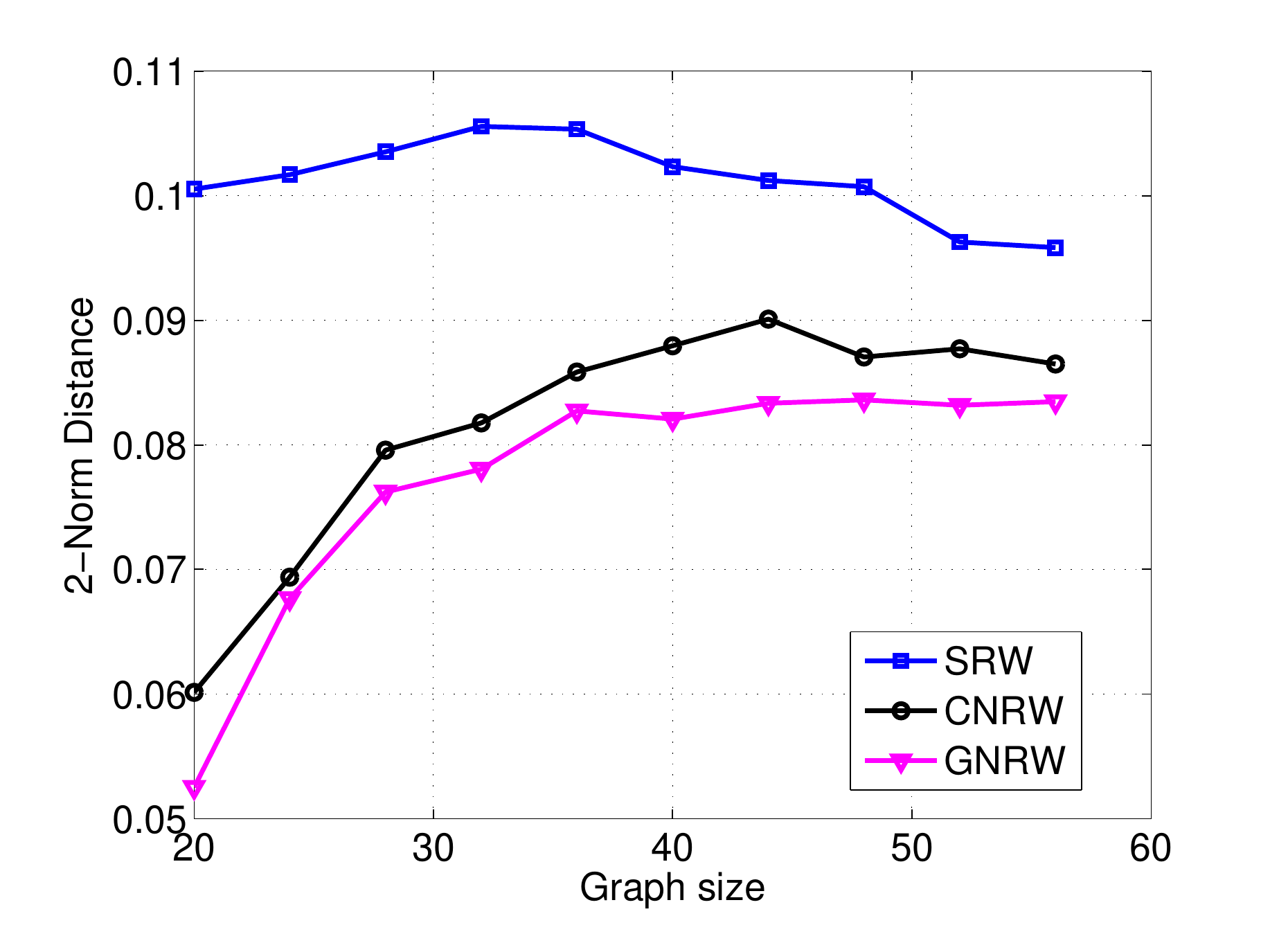}
  \label{fig:barbell_size_2norm}
  }
  \subfigure[Estimation error]{
  \includegraphics[width=.31\textwidth]{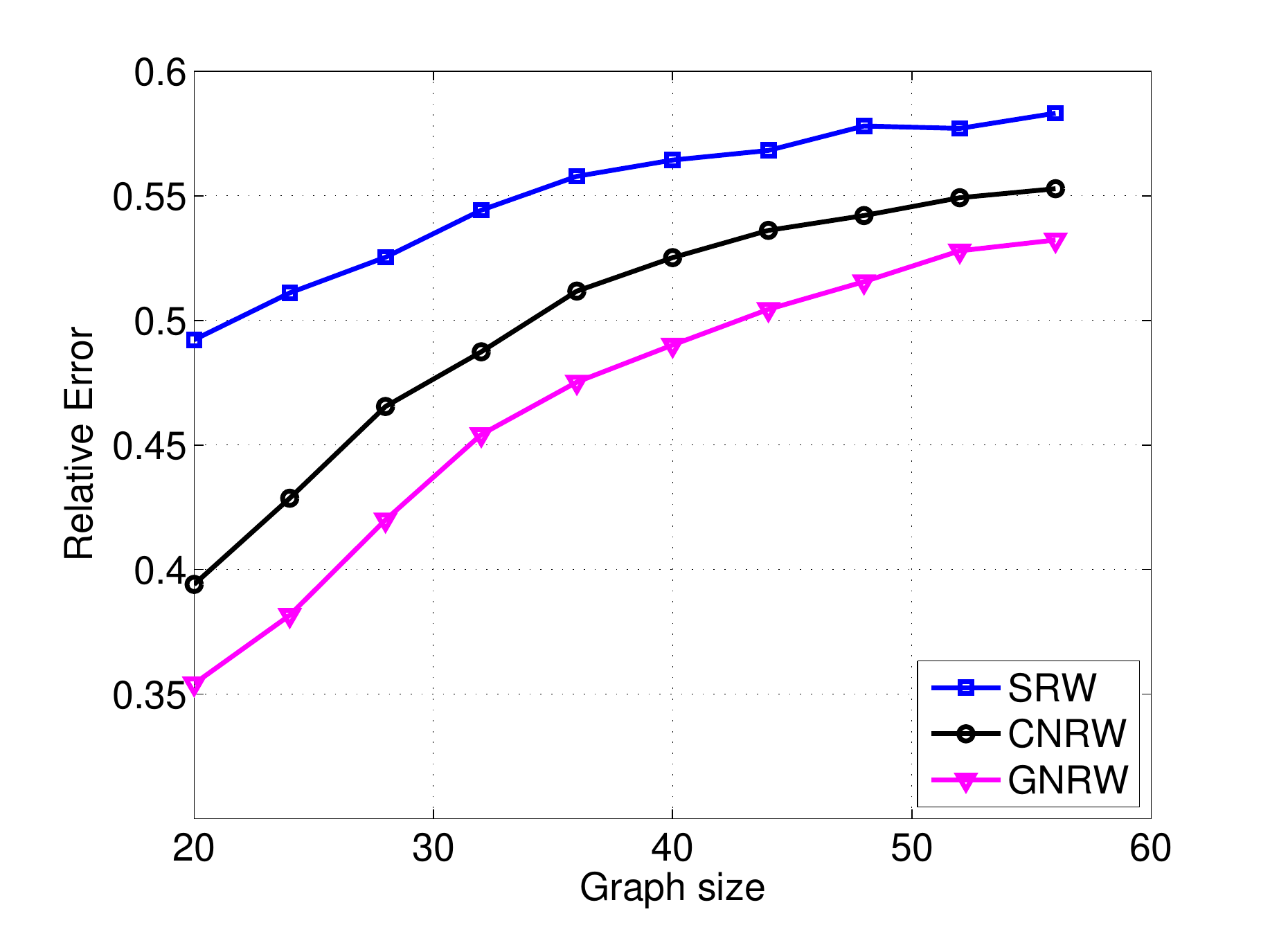}
  \label{fig:barbell_size_est}
  }
\caption{Barbell graph size analytics.}
\label{fig:barbell-size}
\end{figure*}

\section{Related Work}
\label{sec:ref}
\noindent\textbf{Sampling from online social networks.}
Online social networks are different from other data mining resources because of their limited access interface, thus a lot of papers like \cite{Leskovec2006a, Airoldi:2005:SAP:1117454.1117457, Kurant} targeted the challenges of how to efficiently sample from large graphs. 

With some extend of the global topology (e.g. the id range of all the nodes in the graph), \cite{Leskovec2006a} summarized  sampling techniques including random node sampling, random edge sampling and random subgraph sampling in large graphs. \cite{Jin} combines random jump and MHRW together to get efficient uniform samples. \cite{Ribeiro:2010:ESG:1879141.1879192} also demonstrated the frontier sampling that leveraged the advantage of having uniform random initial nodes. Without global topology, \cite{Leskovec2006a} compared sampling techniques such as Simple Random Walk, Metropolis-Hastings Random Walk and traditional Breadth First Search (BFS) and Depth First Search (DFS). \cite{Gjoka2010} confirmed that MHRW is less efficient than SRW because MHRW mixes slower. \cite{58571882} introduced non-backtracking random walk. Also \cite{Alon2008} considered many parallel random walks at the same time. 

Our work extends random walks to higher order MCMCs and systematically consider the historical information by introducing path blocks, which is fundamentally different from existing techniques.

\noindent\textbf{Theoretical analyses of the random walk path blocks}. According to \cite{Neal2004}, the stratification of a random walk's path blocks can affect the asymptotic variance of its estimation. \cite{58571882} also applied \cite{Neal2004}'s theorem to show that non-backtracking random walk is always better than SRW. Our work is based on the construction of the path blocks, and we further discussed about how to design the random walk to make it as a better form of the stratification for the path blocks.

\section{Conclusions}
\label{sec:con}
In this paper, we considered a novel problem of leveraging historic transitions in the design of a higher-ordered MCMC random walk technique, in order to enable more efficient sampling of online social networks that feature restrictive access interfaces. Specifically, we developed two algorithms: (1) CNRW, which replaces the memoryless transition in simple random walk with a memory-based, sampling-without-replacement, transition design, and (2) GNRW, which further considers the observed attribute values of neighboring nodes in the transition design. We proved that while CNRW and GNRW achieve the exact same target (sampling) distribution as traditional simple random walks, they offer provably better (or equal) efficiency no matter what the underlying graph topology is. We also demonstrated the superiority of CNRW and GNRW over baseline and state-of-the-art sampling techniques through experimental studies on multiple real-world online social networks as well as synthetic graphs.

\balance

\bibliographystyle{abbrv}
\bibliography{13_Exp}

\end{document}